\documentclass[11pt]{article}
\usepackage{amsfonts,amsmath,amsthm,amssymb,hyperref,cleveref}

\usepackage{comment, enumerate}
\usepackage[margin=1in]{geometry}
\usepackage{soul,color}
\usepackage{mathrsfs}
\usepackage{mathdots}
\usepackage[linesnumbered,boxed,ruled,vlined]{algorithm2e}

\usepackage{bbm}

\usepackage{url}

\usepackage{tabu}

\usepackage{xfrac}

\usepackage[normalem]{ulem}

\newcommand{\poly}{\mathop{\mathrm{poly}}}

\newcommand{\R}{\mathbb{R}}
\newcommand{\Z}{\mathbb{Z}}

\newcommand{\eps}{\varepsilon}

\theoremstyle{plain}\newtheorem{theorem}{Theorem}[section]
\newtheorem{lemma}[theorem]{Lemma}
\newtheorem{proposition}[theorem]{Proposition}
\newtheorem{corollary}[theorem]{Corollary}

\newtheorem{fact}[theorem]{Fact}
\theoremstyle{definition}
\newtheorem{definition}[theorem]{Definition}
\newtheorem{remark}[theorem]{Remark}
\newtheorem{problem}[theorem]{Problem}

\usepackage[utf8]{inputenc}

\title{Optimal-Degree Polynomial Approximations for Exponentials and Gaussian Kernel Density Estimation} 

\author{Amol Aggarwal\footnote{Department of Mathematics at Columbia University and Institute for Advanced Study. \url{amolaggarwal@math.columbia.edu}. Partially supported by NSF grants DGE-1144152 and DMS-1664619, a Harvard Merit/Graduate Society Term-time Research Fellowship, and a Clay Research Fellowship.
} \and Josh Alman\footnote{Department of Computer Science at Columbia University. \url{josh@cs.columbia.edu}. Partially supported by a Harvard Michael O. Rabin postdoctoral fellowship.}}

\begin{document}

\maketitle

\begin{abstract}
For any real numbers $B \ge 1$ and $\delta \in (0, 1)$ and function $f: [0, B] \rightarrow \mathbb{R}$, let $d_{B; \delta} (f) \in \mathbb{Z}_{> 0}$ denote the minimum degree of a polynomial $p(x)$ satisfying $\big| p(x) - f(x) \big| < \delta$ for each $x \in [0, B]$. In this paper, we provide precise asymptotics for $d_{B; \delta} (e^{-x})$ and $d_{B; \delta} (e^{x})$ in terms of both $B$ and $\delta$, improving both the previously known upper bounds and lower bounds. In particular, we show that

$$d_{B; \delta} (e^{-x}) = \Theta\left( \max \left\{ \sqrt{B \log(\delta^{-1})}, \frac{\log(\delta^{-1}) }{ \log(B^{-1} \log(\delta^{-1}))} \right\}\right), \text{ and}$$

$$d_{B; \delta} (e^{x}) = \Theta\left( \max \left\{ B, \frac{\log(\delta^{-1}) }{ \log(B^{-1} \log(\delta^{-1}))} \right\}\right),$$

\noindent and we explicitly determine the leading coefficients in most parameter regimes.

Polynomial approximations for $e^{-x}$ and $e^x$ have applications to the design of algorithms for many problems, including in scientific computing, graph algorithms, machine learning, and statistics. Our degree bounds show both the power and limitations of these algorithms.

We focus in particular on the Batch Gaussian Kernel Density Estimation problem for $n$ sample points in $\Theta(\log n)$ dimensions with error $\delta = n^{-\Theta(1)}$. We show that the running time one can achieve depends on the square of the diameter of the point set, $B$, with a transition at $B = \Theta(\log n)$ mirroring the corresponding transition in $d_{B; \delta} (e^{-x})$:
\begin{itemize}
    \item When $B=o(\log n)$, we give the first algorithm running in time $n^{1 + o(1)}$. 
    \item When $B = \kappa \log n$ for a small constant $\kappa>0$, we give an algorithm running in time $n^{1 + O(\log \log \kappa^{-1} /\log \kappa^{-1})}$. The $\log \log \kappa^{-1} /\log \kappa^{-1}$ term in the exponent comes from analyzing the behavior of the leading constant in our computation of $d_{B; \delta} (e^{-x})$.
    \item When $B = \omega(\log n)$, we show that time $n^{2 - o(1)}$ is necessary assuming SETH.
\end{itemize}

\end{abstract}
\thispagestyle{empty}
\newpage
\setcounter{page}{1}

\section{Introduction}

Polynomial approximations of important functions play a key role in many areas of computer science and mathematics. We measure the extent to which a function can be approximated by a degree $d$ polynomial as follows.

\begin{definition}

\label{dbdeltaf}

For any real numbers $B \ge 1$ and $\delta \in (0, 1)$, and function $f: [0, B] \rightarrow \mathbb{R}$, let $d_{B; \delta} (f) \in \mathbb{Z}_{> 0}$ denote the minimum degree of a non-constant polynomial $p(x)$ satisfying
\begin{flalign*}
\displaystyle\sup_{x \in [0, B]} \big| p(x) - f(x) \big| < \delta.
\end{flalign*}

\end{definition}

Past work in polynomial approximation theory has typically focused on the case when $B = O(1)$; see, for example, \cite[Chapter 7]{TAFRV}. However, recent computer science applications have motivated studying the setting where both $B$ and $\delta^{-1}$ are growing simultaneously. Indeed, in algorithmic applications, both the magnitude of the input to the function $f$ and the tolerance for error can scale with the size of the input to the problem.

In this paper, we focus specifically\footnote{We mention, however, that the method used in this paper is quite general and is expected to more broadly apply for functions $f$ whose Taylor series coefficients decay sufficiently quickly.} on the functions $e^x$ and $e^{-x}$. As we will discuss more shortly, polynomial approximations for these functions appear naturally in computational problems throughout scientific computing, graph algorithms, machine learning, statistics, and many other areas.
Precisely determining $d_{B; \delta} (e^{-x})$ and $d_{B; \delta} (e^{x})$ is particularly important since in a number of algorithmic applications, such as the batch Gaussian Kernel Density Estimation that we discuss in Section~\ref{section:gKDE} below, these quantities appear in the exponent of the input size in the running time. In these settings, logarithmic or even constant factors can be the difference between a fast or a trivially slow running time (see especially Sections~\ref{subsec:1} and \ref{subsec:3} below). The standard framework of approximation theory (e.g.,~\cite{Pow,TAFRV,AT}) can be used to deduce bounds on $d_{B; \delta}$ that are typically suboptimal, often losing (at least) such logarithmic factors, especially in the regime when $B$ is large.

Our main results are tight asymptotics, including the exact leading constant in most parameter regimes (see \Cref{rem6} below), for both $d_{B; \delta} (e^{-x})$ and $d_{B; \delta} (e^{x})$. 

In what follows, we define the function 
\begin{flalign}
\label{gx} 
G (x) = \sqrt{x^2 + 1} + x \log \big( \sqrt{x^2 + 1} - x \big),
\end{flalign} 

\noindent for each $x \in \mathbb{R}_{\ge 0}$. 

\begin{theorem}[Approximate degree of $e^{-x}$] 

\label{estimatedegree2}

Let $B \ge 1$ and $\delta \in (0, 1)$. Then, $$d_{B; \delta} (e^{-x}) = \Theta\left( \max \left\{ \sqrt{B \log(\delta^{-1})}, \frac{\log(\delta^{-1}) }{ \log(B^{-1} \log(\delta^{-1}))} \right\}\right).$$ More precisely, we have the following asymptotics as $B + \delta^{-1}$ tends to $\infty$.

\begin{enumerate}

{\item If $B = o \big( \log ( \delta^{-1} ) \big)$, then $d_{B; \delta} (e^{-x})  =  \bigg( \displaystyle\frac{\log \big( \delta^{-1} \big)}{\log \big( B^{-1} \log (\delta^{-1}) \big)} \bigg) \big( 1 + o (1) \big)$.}

{\item If $B = 2 r \log ( \delta^{-1})$ for fixed $r > 0$, then $d_{B; \delta} (e^{-x}) = \big( \nu r + o(1) \big) \log (\delta^{-1})$, where $\nu = \nu (r) > 0$ is the unique positive solution\footnote{The uniqueness of this solution (and the ones to be mentioned below) follows from the facts that $G (0) = 1$, $\lim_{z \rightarrow \infty} G(z) = -\infty$, and $G' (z) < 0$ for $z > 0$.} to the equation $G (\nu) = 1 - r^{-1}$.}

{\item If $B = \omega \big( \log ( \delta^{-1} ) \big)$ and $B \le \delta^{-o(1)}$, then $d_{B; \delta} (e^{-x}) = \big( 1 + o(1) \big) \sqrt{B \log (\delta^{-1})}$.}

{\item If $B \ge \delta^{- \Omega (1)}$, then $d_{B; \delta} (e^{-x}) = \Theta \left( \sqrt{B \log \big( \delta^{-1} \big)} \right)$.}

\end{enumerate}

\end{theorem}

\begin{theorem}[Approximate degree of $e^{x}$] 

\label{estimatedegree1}

Let $B \ge 1$ and $\delta \in (0, 1)$. Then,
$$d_{B; \delta} (e^{x}) = \Theta\left( \max \left\{ B, \frac{\log(\delta^{-1}) }{ \log(B^{-1} \log(\delta^{-1}))} \right\}\right).$$ More precisely, we have the following asymptotics as $B + \delta^{-1}$ tends to $\infty$.

\begin{enumerate}

{\item If $B = o \big( \log ( \delta^{-1}) \big)$, then $d_{B; \delta} (e^x)  =  \bigg( \displaystyle\frac{\log \big( \delta^{-1} \big)}{\log \big( B^{-1} \log (\delta^{-1}) \big)} \bigg) \big( 1 + o (1) \big)$.}

{\item If $B = 2 r \log ( \delta^{-1})$ for fixed $r > 0$, then $d_{B; \delta} (e^x) = \big( \mu r + o(1) \big) \log (\delta^{-1})$, where $\mu = \mu (r) > 0$ is the unique positive solution to the equation $G (\mu) = -1 - r^{-1}$.}

{\item If $B = \omega \big( \log \big( \delta^{-1}) \big)$, then $d_{B; \delta} (e^x)  = \displaystyle\frac{z_*B}{2} \big( 1 + o(1) \big)$, where $z_* \approx 2.2334$ denotes the unique positive solution to the equation $G (z_*) = -1$.}

\end{enumerate}

\end{theorem}

\begin{remark}
Polynomials achieving the degree upper bounds stated in Theorems~\ref{estimatedegree2} and \ref{estimatedegree1} can be constructed in $\poly(d)$ time, with coefficients which are rational numbers with $\poly(d)$-bit integer numerators and denominators, where $d$ is the degree.
\end{remark}

\begin{remark} \label{rem6}

In the fourth case of \Cref{estimatedegree2}, we do not determine the leading constant $A = A(B; \delta)$ for which $d_{B; \delta} (e^{-x}) = \big( A + o(1) \big) \sqrt{B \log \big( \delta^{-1} \big)}$; as we will see below, when $\delta = o(1)$, we only bound it between $\frac{1}{2} \leq A \leq 1$. It is unclear to us whether or not this constant would admit a concise description in this parameter regime, especially in the case when $\delta$ is fixed as $B$ tends to $\infty$. 
In all other parameter regimes of \Cref{estimatedegree2}, and in every case of \Cref{estimatedegree1}, we determine the exact leading constant.\footnote{It is quickly verified that the constants $\nu$, $\mu$, and $z_*$ from \Cref{estimatedegree2} and \Cref{estimatedegree1} satisfy $2r^{-1/2} \le \nu \le \max \big\{ r^{-1}, e \big\}$ and $z_* \le \mu \le \max \big\{ r^{-1}, e \big\}$ for all $r>0$.}

\end{remark} 

 \textbf{Previous bounds.} The question of providing tight bounds on $d_{B; \delta} (e^{-x})$ was posed in works of Orecchia, Sachdeva, and Vishnoi \cite[Sections 4 and 7]{orecchia2012approximating}, and Sachdeva and Vishnoi \cite[Section 5]{sachdeva2014faster}. They were motivated by algorithmic applications, as \cite{orecchia2012approximating} showed how upper bounds on $d_{B; \delta} (e^{-x})$ can be used to design faster algorithms for the Balanced Separator problem from spectral graph theory. They gave an upper bound of $d_{B; \delta} (e^{-x}) \leq O(\sqrt{\max \{ \log (\delta^{-1}), B \}} \cdot \log^{3/2} (\delta^{-1}))$, and a lower bound of $d_{B; \delta} (e^{-x}) \geq \frac12 \sqrt{B}$. Later, \cite{sachdeva2014faster} improved the upper bound to $d_{B; \delta} (e^{-x}) \leq O(\sqrt{\max \{ \log (\delta^{-1}), B \}} \cdot \log^{1/2} (\delta^{-1}))$ (noting that such a bound was also implicit in \cite{hochbruck1997krylov}). Theorem~\ref{estimatedegree2} provides precise asymptotics for $d_{B; \delta} (e^{-x})$, thereby answering the above question.
 
 In particular, Theorem~\ref{estimatedegree2} shows that the prior upper bound could be improved by a logarithmic factor in some parameter regimes, but was otherwise asymptotically tight. For the Balanced Separator problem studied by~\cite{orecchia2012approximating}, where the running time depends polynomially on~$d_{B; \delta} (e^{-x})$, this rules out a big improvement without a new approach. For other applications where the running time has an exponential dependence on $d_{B; \delta} (e^{-x})$, our improvements have more significant implications. For instance, as we discuss below in Section~\ref{subsec:1}, in some parameter regimes of the batch Gaussian Kernel Density Estimation problem, our Theorem~\ref{estimatedegree2} yields a near linear time algorithm, whereas applying instead the prior bound of~\cite{sachdeva2014faster} would only yield a trivial quadratic running time.
 
 We are unaware of prior work which specifically bounded $d_{B; \delta} (e^{x})$, although one could apply standard results on Chebyshev interpolation (such as~\cite[{Theorem 8.2}]{AT}) with some work to yield a bound $d_{B; \delta} (e^{x}) \geq \Omega(\max\{B, \log(\delta^{-1}) \})$.\\

 \textbf{Phase transitions.} In fact, Theorem \ref{estimatedegree2} and Theorem \ref{estimatedegree1} indicate that the dependence of the optimal degrees $d_{B; \delta} (e^{-x})$ and $d_{B; \delta} (e^{x})$ on the parameters $B$ and $\delta$ is quite intricate. First, their orders of magnitudes both exhibit transitions depending on the relative sizes of $B$ and $\log (\delta^{-1})$. For example, when $B = \omega \big( \log (\delta^{-1}) \big)$, \Cref{estimatedegree2} shows that $d_{B; \delta} (e^{-x})$ exhibits square root dependence on both $\log (\delta^{-1})$ and $B$, but when $B = o \big( \log (\delta^{-1}) \big)$ it exhibits nearly linear dependence on $\log (\delta^{-1})$ and only logarithmic dependence on $B$. Second, in the ``critical regime'' $B = 2r \log (\delta^{-1})$, \Cref{estimatedegree2} shows that $d_{B; \delta} (e^{-x}) = \Theta \big( \log (\delta^{-1}) \big)$, whose implicit constant is obtained by solving the transcendental equation $G(z) = 1 - r^{-1}$. A similar transition (with a transcendental leading constant in the critical regime) is shown for the approximating degrees of $e^x$ in \Cref{estimatedegree1}, but with the qualitative difference that $d_{B; \delta} (e^x)$ is linear in $B$ (to leading order) and independent of $\delta$, for $B = \omega \big( \log (\delta^{-1}) \big)$. 

To our knowledge, this is the first appearance of a transition arising when one simultaneously scales $B$ and $\delta$ in the context of polynomial approximation theory. Indeed, as mentioned previously, prior works in this direction typically analyzed the case $B = O (1)$, where transitions like these are not visible. As we will explain in \Cref{Algorithm} below, these behaviors for $d_{B; \delta} (e^{-x})$ and  $d_{B; \delta} (e^{x})$ will have algorithmic interpretations. For example, we will see that the estimates on  $d_{B; \delta} (e^{-x})$ provided in \Cref{estimatedegree2} imply a fine-grained computational phase transition for Gaussian Kernel Density Estimation in certain parameter regimes. \\

\textbf{Previous methods.} In the theoretical computer science literature, proofs of upper and lower bounds on the approximate degree $d_{B; \delta} (f)$ of a function $f$ had been typically based on two distinct arguments \cite{nisan1994degree,shi2002approximating, aaronson2004quantum, ambainis2005polynomial, orecchia2012approximating, sachdeva2014faster, bun2015dual}.  Upper bounds were often shown by providing an explicit polynomial approximation for $f$, usually given by (a truncation of) the expansion of $f$ in the basis of Chebyshev polynomials. Lower bounds were typically shown by making use of an estimate, such as Markov Brothers' inequality, that constrains the maximum derivative of a bounded polynomial in terms of its degree. Both ideas are archetypes of classical approximation theory; see \cite[{Chapters 2 and 4}]{TAFRV}. \\

 \textbf{Our methods.} As above, to upper bound $d_{B; \delta} (f)$ we will explicitly provide an approximating polynomial for $f$, obtained from the Chebyshev expansion of its rescale $f_B (x) = f \big( \frac{B}{2} (1 - x) \big)$ (whose domain is now $[-1, 1]$). However, derivative-degree estimates such as Markov's inequality that prior works used to lower bound $d_{B; \delta} (f)$ usually become insensitive to the tolerance parameter $\delta$ once it passes below a (typically non-optimal) threshold. Thus, they will not suffice for our purposes of pinpointing the precise asymptotic behavior of $d_{B; \delta} (f)$.

We therefore proceed differently, by instead again making use of the Chebyshev expansion of $f_B (x)$. In particular, we use the orthogonality of the Chebsyhev polynomials to lower bound the minimal distance from $f_B$ to a polynomial of degree $d$ in terms of the series coefficients of $f_B$ when expanded in the Chebyshev basis; see \Cref{ajcj} below. Thus, bounds on these series coefficients can be used to bound $d_{B; \delta} (f)$. This idea was also ubiquitous in the traditional theory and practice of approximating polynomials; for instance, it was very fruitful in proving the classical sharp estimates \cite[{Chapter 7.8 (22)}]{TAFRV} on $d_{B; \delta} (e^{-x})$ and $d_{B; \delta} (e^{-x})$ when $B = O(1)$. 

However, to our understanding, this idea has not been implemented before in our context where $B$ and $\delta$ scale jointly (either in the computer science or approximation theory literature). In this setting, we must study the limiting behaviors for the high-degree coefficients in the Chebyshev expansion $f_B$, simultaneously as the degree and as $B$ tend to $\infty$; see \Cref{avlambdalimit} below. This analysis becomes more involved than in the case $B = O(1)$, as it should in order to give rise to the intricate asymptotic phenomena described in \Cref{estimatedegree2} and \Cref{estimatedegree1}. In particular, the phase transitions observed in those results can be traced to corresponding phase transitions for these series coefficients, given in \Cref{psivlambdakappa} below.

\subsection{Algorithmic Applications}

\label{Algorithm}

Polynomial approximations with low error have numerous applications throughout algorithm design and complexity theory; see, for instance, the introduction of the survey by Sachdeva and Vishnoi~\cite{sachdeva2014faster} for an overview. 
The quantities $d_{B; \delta} (e^{-x})$ and $d_{B; \delta} (e^{x})$, in particular, play a central role in many algorithms due to the prevalence of exponential functions. Some examples include:
\begin{itemize}
    \item \textbf{Approximating matrix exponentials.} Given a matrix $A$ and a vector $v$, approximate $e^A \cdot v$. One of the most common algorithms in theory and in practice for this problem is the Lanczos method~\cite{lanczos1950iteration}, whose running time is bounded by $O(d \cdot m_A + d^2)$~\cite{musco2018stability}, where $m_A$ is the amount of time required to do a matrix-vector multiplication by $A$, and $d = d_{B; \delta} (e^{x})$ is the approximate degree which we compute in~\Cref{estimatedegree1} with $B = ||A||$ and $\delta$ is the desired approximation error parameter.
    \item \textbf{Finding balanced separators in graphs.} The aforementioned work by Orecchia, Sachdeva and Vishnoi \cite{orecchia2012approximating} uses polynomial approximations of $e^{-x}$ in a way similar to the Lanczos method to give fast, practical algorithms for the Balanced Separator problem.\footnote{They also give a faster algorithm in some special cases using \emph{rational approximations} of $e^{-x}$.}
    \item \textbf{Estimating softmax.} Many ``multinomial classification'' problems in natural language processing and other areas make use of the \emph{softmax} function to convert vectors representing the different classes into estimated probabilities. Given $n$ vectors $w_1, \ldots, w_n \in \R_{\geq 0}^m$, an index $i \in [n]$ and a sample vector $h \in \R_{\geq 0}^m$, softmax is defined as $$\text{softmax}(h,i,w_1,\ldots,w_n) := \frac{e^{\langle w_i, h \rangle}}{\sum_{j=1}^n e^{\langle w_j, h \rangle}}.$$ Training models in these applications frequently requires many softmax computations, and so approximations of softmax which are faster to compute are often used~\cite{chen2015strategies}. Replacing the exponentials in softmax by the optimal polynomial approximations we give in~\Cref{estimatedegree1} can be used to more quickly compute such approximations~\cite{joulin2017efficient,nilsson2014hardware}.
    \item \textbf{Kernel methods.} Polynomial approximations for $e^{-x}$ have been used to design faster sketching and estimation techniques for Gaussian kernels, including in a number of recent algorithms; see e.g.~\cite{yang2003improved,lee2019finding,alman2020algorithms,ahle2020oblivious}. In~\Cref{section:gKDE} below, we show a new application along these lines to batch Gaussian Kernel Density Estimation.
\end{itemize}

\subsection{Gaussian Kernel Density Estimation} \label{section:gKDE} Kernel Density Estimation (KDE) is one of the most common methods for non-parametric estimation of the density of an unknown distribution $\cal{D}$. Given a set $P \subset \R^m$ of samples from $\cal{D}$, along with a weight $w_y \in \R$ for each $y \in P$, the kernel density function (KDF) of $P$ at a point $x \in \R^m$ is given by $$KDF_P(x) := \sum_{y \in P} w_y \cdot k(x,y),$$ where $k : \R^m \times \R^m \to \R$ is a carefully chosen \emph{kernel function}. In most applications, one would like to compute $KDF_P$ at many points $x$. Perhaps the most commonly studied kernel function is the \emph{Gaussian kernel} $k(x,y) = e^{-\|x-y\|_2^2}$. This motivates the question:
    
\begin{problem}[Batch Gaussian KDE]
Given as input $2n$ points $x^{(1)}, \ldots, x^{(n)}, y^{(1)}, \ldots, y^{(n)} \in \R^m$ which implicitly define the matrix $K \in \R^{n \times n}$ by $K[i,j] = e^{-\|x^{(i)} - y^{(j)}\|_2^2}$, as well as a vector $w \in \R^n$, and an error parameter $\delta>0$, compute an approximation to $K \cdot w$, meaning, output a vector $v \in \R^n$ such that $\|K\cdot w - v\|_\infty \leq \delta \cdot \|w\|_1$.
\end{problem}

\textbf{Polynomial method algorithm.} This problem can be solved by using a polynomial approximation to $e^{-x}$ in order to construct a low-rank approximation to the matrix $K$, as follows. Let $B \ge 1$ and $\delta \in (0, 1)$ denote real numbers, and suppose $p(z)$ is a univariate polynomial of degree $d \geq d_{B; \delta} (e^{-x})$ such that 
\begin{flalign*}
\displaystyle\sup_{z \in [0, B]} \big| p(z) - e^{-z} \big| \le \delta.
\end{flalign*}
Thus, for $x,y \in \R^m$ with $\|x-y\|_2^2 \leq B$, the polynomial $p(\sum_{\ell=1}^m (x_\ell - y_\ell)^2)$ outputs a value within an additive $\delta$ of $e^{-\|x-y\|_2^2}$. Hence, to solve Batch Gaussian KDE, it suffices to output the vector $\tilde{K} \cdot w$, where $\tilde{K} \in \R^{n \times n}$ is the matrix given by $\tilde{K}[i,j] = p(\sum_{\ell=1}^m (x_\ell^{(i)} - y_\ell^{(j)})^2)$. By a standard argument, the rank of $\tilde{K}$ is at most the number of monomials in the expansion of $p(\sum_{\ell=1}^m (x_\ell - y_\ell)^2)$, which is bounded above by $M \leq \binom{2d+2m}{2d}$, and the corresponding low rank expression for $\tilde{K}$ can be found in time $O(n \cdot M \cdot m)$.

In other words, whenever $M < n^{o(1)}$, we can solve Batch Gaussian KDE in deterministic $n^{1+o(1)}$ time in this way (see also~\cite[Section~5.3]{alman2020algorithms} where this approach was previously laid out). Theorem~\ref{estimatedegree2} characterizes exactly when this is possible in terms of $m$ (the dimension of the points), $B$ (the square of the diameter of the point set), and $\delta$ (the error parameter):

\begin{corollary} \label{kdecorr}
For any positive integer $m  < n^{o(1)}$, and real numbers $B \geq 1$ and $\delta \in (0,1)$, define $d = d_{B; \delta} (e^{-x})$ as in Theorem~\ref{estimatedegree2}. 
Then, batch Gaussian KDE can be solved in deterministic time $n^{1+o(1)}$ whenever $\binom{2d+2m}{2d} < n^{o(1)}$.
Similarly, if $\binom{2d+2m}{2d} < n^{c}$ for some constant $0 < c < 1$, then batch Gaussian KDE can be solved in truly subquadratic deterministic time $n^{1 + c + o(1)}$.
\end{corollary}

\subsubsection{Comparison with prior work.}\label{subsec:1} The previous best known algorithm for Batch Gaussian KDE is due to recent work of Charikar and Siminelakis~\cite{charikar2017hashing}, which showed how to solve this problem in randomized time $\delta^{-2} n^{1+o(1)} \cdot (\log n)^{O(B^{2/3})}$ for any dimension $m < n^{o(1)}$. Their algorithm achieves randomized running time $n^{1+o(1)}$ whenever $\delta^{-1} < n^{o(1)}$, $B < o((\log n / \log\log n)^{3/2})$, and $m < n^{o(1)}$. 

Focusing on the setting\footnote{Often, depending on the desired error guarantees, one can reduce to roughly this case using dimensionality reduction like the Johnson–Lindenstrauss lemma.} where $m = O(\log n)$, Corollary~\ref{kdecorr} achieves deterministic running time $n^{1+o(1)}$ in all the same parameter settings as the previous algorithm, and also new settings including:
\begin{itemize}
    \item When $B = o(\log^2 n)$ and $\delta^{-1} < n^{o(\log n / B)}$ (slightly improving the parameter $B$), or
    \item When $B = o(\log n)$ and $\delta^{-1} = n^{\Theta(1)}$ (considerably improving the parameter $\delta$).
\end{itemize}

In particular, the latter setting enables us to take $\delta$ to depend polynomially in $n$, while still retaining an $n^{1+o(1)}$ running time.

The above parameter regimes are also where our upper bound on $d_{B; \delta} (e^{-x})$ from Theorem~\ref{estimatedegree2} logarithmically improves on the one given in~\cite{sachdeva2014faster}. This improvement was in fact necessary for our application to KDE, as the estimate from \cite{sachdeva2014faster} was $d_{B; \delta} (e^{-x}) = \Omega (\log n)$ in these settings, which would give a running time of $n \cdot \binom{2d+2m}{2d} \geq n^{1 + \Omega(1)}$, as opposed to our near-linear one.

Interestingly, our algorithm and that of~\cite{charikar2017hashing} take approaches which rely on very different properties of the kernel function $k$. Charikar and Siminelakis's algorithm uses a clever Locality-Sensitive Hashing-based approach, and also works well for other kernels with efficient hash functions, whereas our approach instead requires $k$ to have a low-degree polynomial approximation. Other popular algorithmic techniques for KDE, such as the Fast Multipole Method~\cite{greengard1987fast}, or core-sets~\cite{agarwal2005geometric,phillips2013varepsilon}, lead to $n^{1 + \Omega(1)}$ running times in the high-dimensional $d = \Omega(\log n)$, low-error $\eps < n^{-\Theta(1)}$ setting; see~\cite[Section~1.3.2]{syminelakis2019fast} for an overview of these known approaches.  \\

 \subsubsection{SETH lower bound.}\label{subsec:2} To complement \Cref{kdecorr}, we also show a fine-grained lower bound, that assuming the Strong Exponential Time Hypothesis (SETH), when $m = \Theta(\log n)$ and $\delta^{-1} = n^{\Theta(1)}$, one cannot achieve running time $n^{1+o(1)}$ when $B = \Omega(\log n)$.

\begin{proposition} \label{sethlb}
Assuming SETH, for every $q>0$, there are constants $\alpha, \beta, \kappa > 0$ such that Batch Gaussian KDE in dimension $m = \alpha \log n$ and error $\delta = n^{-\beta}$ for input points whose diameter squared is at most $B = \kappa \log n$ requires time $\Omega(n^{2-q})$.
\end{proposition}

The proof of \Cref{sethlb} is a slight modification of a similar lower bound of Backurs, Indyk, and Schmidt~\cite{backurs2017fine}, which relates Gaussian KDE to nearest neighbor search (for which SETH lower bounds are already known~\cite{rubinstein2018hardness}). 

To summarize, in the natural setting where $m = \Theta(\log n)$ and $\delta^{-1} = n^{\Theta(1)}$:
\begin{itemize}
    \item Our algorithm using the polynomial method achieves running time $n^{1+o(1)}$ when $B = o(\log n)$. 
    \item Assume SETH. It is not possible to improve our algorithm to achieve running time $n^{1 + o(1)}$ when $B = \Theta (\log n)$. Moreover, if $B = \omega (\log n)$, then no algorithm achieves running time faster than $n^{2 - o(1)}$.
    
\end{itemize}

 \subsubsection{Critical regime behavior from the leading constant in~\Cref{estimatedegree2}.}\label{subsec:3} Thus, assuming SETH (and under the setting $m = \Theta (\log n)$ and $\delta^{-1} = n^{\Theta (1)}$), the complexity of Batch Gaussian KDE exhibits a transition mirroring the one exhibited by $d_{B; \delta} (e^{-x})$ in \Cref{estimatedegree2}. More specifically, as $B$ goes from $o (\log n)$ to $\omega (\log n)$, this complexity transitions from $n^{1 + o(1)}$ to $n^{2 - o(1)}$. In the ``critical regime'' where $B = \kappa \log n$ for some fixed $\kappa > 0$, this suggests that its complexity should grow as $n^{1 + \varphi (\kappa)}$, for some non-decreasing function $\varphi: \mathbb{R}_{> 0} \rightarrow \mathbb{R}_{> 0}$ satisfying $\lim_{\kappa \rightarrow 0} \varphi (\kappa) = 0$ and $\lim_{\kappa \rightarrow \infty} \varphi(\kappa) = 1$. It would be fascinating to better understand more precise behavior of this function $\varphi$. Does it continuously transition from $0$ to $1$ as $\kappa$ increases, or does it admit a sudden ``jump'' at a specific threshold value for $\kappa$?

While these questions remain open, we can use our asymptotics for $d_{B; \delta} (e^{-x})$ to provide bounds on $\varphi (\kappa)$ for small $\kappa$. In particular, the below corollary implies that $\varphi (\kappa) = O  \big( \log \log \kappa^{-1} / \log \kappa^{-1} \big)$. Our derivation of the term $\log \log \kappa^{-1} / \log \kappa^{-1}$ appearing in the exponent makes use of the leading constant $\nu$ (defined in the second part of \Cref{estimatedegree2}) for the asymptotics of $d_{B; \delta} (e^{-x})$. Indeed, this is the case of Corollary~\ref{kdecorr} where $d = d_{B; \delta} (e^{-x}) = \Theta(\log n)$ and $m = O(\log n)$, and so $\binom{2d+2m}{2d} = n^{\Theta(1)}$, where the leading constant in $d_{B; \delta} (e^{-x})$ determines the value of the $\Theta(1)$. 

\begin{corollary}

\label{fkappa} 

Fix constants $\alpha, \beta > 0$, and suppose that $m = \alpha \log n$ and $\delta = n^{-\beta}$. If $B = \kappa \log n$ for some $\kappa < \frac{1}{2}$, then Batch Gaussian KDE can be solved in time $O(n^{1 + c \log \log \kappa^{-1} / \log \kappa^{-1}})$, where $c = c(\alpha, \beta) > 0$ only depends on $\alpha$ and $\beta$.

\end{corollary}

Prior work has shown similar ``critical regime'' behavior for other problems with SETH-based lower bounds. The Orthogonal Vectors problem for $n$ vectors in dimension $\kappa \log n$ for large $\kappa$ can be solved in time $n^{2 - 1/O(\log \kappa)}$~\cite{abboud2014more,chan2016deterministic}, whereas the problem in dimension $\omega(\log n)$ requires time $n^{2 - o(1)}$ assuming SETH. The Batch Hamming Nearest Neighbors problem for $n$ vectors in dimension $\kappa \log n$ for large $\kappa$ can be solved in time $n^{2 - 1/\tilde{O}(\sqrt{\kappa})}$~\cite{alman2015probabilistic,alman2016polynomial}, whereas the problem in dimension $\omega(\log n)$ also requires time $n^{2 - o(1)}$ assuming SETH. Interestingly, these algorithms make use of variants on the polynomial method using \emph{probabilistic} polynomials, whereas we make use of approximate polynomials here.

\section{Proof Overview}

In this section we outline the proofs of \Cref{estimatedegree2} and \Cref{estimatedegree1}, which will be established in detail in Section 3 below. To that end, we will use the Chebyshev polynomials, which are defined as follows; see \Cref{sec:prelims} for a more thorough explanation of its properties.

	\begin{definition}
		
		\label{pd1} 
		
		Fix an integer $d \ge 0$. Let $\mathcal{P}_d \subset \mathbb{R} [x]$ denote the set of single-variable polynomials $p(x)$ with $\deg p \le d$, and define the degree $d$ \emph{monic Chebyshev polynomial} $Q_d (x) \in \mathcal{P}_d$ as follows. Set $Q_0 (x) = 1$ and, for each $d \ge 1$, define $Q_d (x)$ by imposing that
	\begin{flalign}
	\label{qdtheta}
	Q_d (\cos \theta) = 2^{1 - d} \cos (d \theta), \qquad \text{for each $\theta \in [0, 2 \pi]$.}
	\end{flalign}
		
	\end{definition}

    It is well understood in the literature that smooth functions $f$ are typically well-approximated by polynomials obtained by truncating the series expansion of $f$ in the basis of Chebyshev polynomials; see, for instance, \cite[(15.5), (15.8)]{AT} for more precise formulations of this statement. 
    
    In particular, the following proposition provides a version of this statement that will be more useful for our purposes. It provides upper and lower bounds on the optimal error of a polynomial approximation $p(x)$ of a function $f: [-1, 1] \rightarrow \mathbb{R}$ in terms of its Chebyshev expansion coefficients. Both bounds in this result are known; the lower bound follows from the $L^2$-orthogonality of the Chebyshev polynomials, and the upper bound follows from \eqref{qdtheta}. Still, we provide a short and self-contained proof of the below proposition in \Cref{ProofSumaj}.

	\begin{proposition}
		
		\label{ajcj}
		
		Let $a_0, a_1, \ldots \in \mathbb{R}$ satisfy $\sum_{j = 0}^{\infty} |a_j| < \infty$. Then, the absolutely convergent series $f: [-1, 1] \rightarrow \mathbb{R}$ defined by $f(x) = \sum_{j = 0}^{\infty} 2^{j - 1} a_j Q_j (x)$ satisfies 
		\begin{flalign}
		\label{pxfx}
		\left( \displaystyle\frac{1}{2} \displaystyle\sum_{k = D}^{\infty} \displaystyle\frac{a_k^2}{k} \right)^{1 / 2} \le \displaystyle\inf_{p \in \mathcal{P}_{D - 1}} \displaystyle\sup_{x \in [-1, 1]} \big| p(x) - f(x) \big| \le \displaystyle\sum_{j = D}^{\infty} |a_j|,
		\end{flalign}
		
		\noindent for any integer $D \ge 1$.

	\end{proposition}
	
	In particular, \eqref{pxfx} provides nearly matching upper and lower bounds (up to a factor of $(2D)^{1/2}$) if the coefficients $\{ a_j \}$ decay sufficiently quickly. Indeed, then the left and right sides of that inequality are asymptotically governed by their leading terms $a_D$. 
	
	As stated, \Cref{ajcj} only applies for approximating polynomials on the interval $[-1, 1]$. However, we would like to approximate $e^{-x}$ and $e^x$ on $[0, B]$, for some $B \ge 1$. Therefore, we rescale by first setting $\lambda = \frac{B}{2}$, and then by observing that to approximate $e^{-x}$ (or $e^x$) on $[0, B]$ it suffices to approximate $e^{\lambda x - \lambda}$ (or $e^{\lambda x + \lambda}$, respectively) on $[-1, 1]$. 
	
	We will show that the coefficients of these latter functions, when written in the Chebyshev basis, indeed decay quickly (with an explicit rate, dependent on $\lambda$), and then apply \Cref{ajcj}. We can in fact compute these coefficients exactly, by first expressing $e^{-x}$ and $e^x$ through their Taylor series, and then by changing basis from the monomials $x^n$ to the Chebyshev polynomials. This yields the following (known) lemma, whose short proof will be recalled in \Cref{Proofavlambdabvlambda} below.

	\begin{lemma} 
		
	\label{amlambda}
	
	For any real numbers $\lambda > 0$ and $x \in [-1, 1]$, we have that 
	\begin{flalign}
	\label{absum} 
	e^{-\lambda x - \lambda} = \displaystyle\sum_{v = 0}^{\infty} 2^{v - 1} A_{v, \lambda} Q_v (x), \qquad e^{\lambda x + \lambda} = \displaystyle\sum_{v = 0}^{\infty} 2^{v - 1} B_{v, \lambda} Q_v (x),
	\end{flalign}
	
	\noindent where for any integer $v \ge 0$ we have set 
	\begin{flalign}
	\label{avlambdadefinition}
	A_{v, \lambda} = 2 e^{-\lambda} (-1)^v \displaystyle\sum_{ n - v \in 2 \mathbb{Z}_{\ge 0}} \displaystyle\frac{\lambda^n}{2^n n!} \binom{n}{\frac{n - v}{2}}, \qquad B_{v, \lambda} = 2 e^{\lambda} \displaystyle\sum_{ n - v \in 2 \mathbb{Z}_{\ge 0}} \displaystyle\frac{\lambda^n}{2^n n!} \binom{n}{\frac{n - v}{2}}.
	\end{flalign}

	\end{lemma}

    We next apply a saddle point analysis to obtain precise asymptotics for $A_{v, \lambda}$ and $B_{v, \lambda}$, as $\lambda + v$ tends to $\infty$. In particular, the following proposition shows that these coefficients decay exponentially in $v$, with an explicit rate function given by $\Psi_{v, \lambda}$ in \eqref{psidefinition}. This exact form of this rate function will eventually serve as the source of the phase transitions for $d_{B; \delta} (e^{-x})$ and $d_{B; \delta} (e^x)$ explained in \Cref{estimatedegree2} and \Cref{estimatedegree1}, respectively. Indeed, one might already observe that $\Psi_{v, \lambda} = \lambda G \big( \frac{v}{\lambda} \big)$, where we recall the function $G(x)$ from those results. The below proposition will be stated a bit informally; we refer to \Cref{avlambdalimit2} below for the more precise formulation needed for our purposes.
	
	\begin{proposition}
		
	\label{avlambdalimit}
		
	Recall the quantities $A_{v, \lambda}$ and $B_{v, \lambda}$ from \eqref{avlambdadefinition} for any integer $v \ge 0$ and real number $\lambda \ge \frac{1}{2}$. Denote
	\begin{flalign}
	\label{psidefinition}
	\Psi_{v, \lambda} = \sqrt{v^2 + \lambda^2} + v \log \bigg( \displaystyle\frac{\sqrt{v^2 + \lambda^2} - v}{\lambda} \bigg).
	\end{flalign}
	
	\noindent As $v + \lambda$ tends to $\infty$, we have that 
		\begin{flalign*}
		(-1)^v A_{v, \lambda} = (\lambda + v)^{O(1)} \exp \big( \Psi_{v, \lambda} - \lambda \big); \qquad  B_{v, \lambda} = (\lambda + v)^{O(1)} \exp \big( \Psi_{v, \lambda} + \lambda \big).
		\end{flalign*}
		
	\end{proposition}

	Combining \Cref{ajcj} and \Cref{avlambdalimit}, we obtain the following corollary, which provides nearly sharp bounds on the error one can achieve for a degree $d$ polynomial approximation of $e^{-x}$ and $e^x$ on $[0, B]$. Once again, the below proposition will be stated a bit informally, and we refer to \Cref{pexponential2} below for a more precise formulation. 
	
    \begin{corollary}
    
    \label{pexponential}
    
    Let $d \ge 1$ be an integer, and let $B \ge 1$ be a real number. Set $\lambda = \frac{B}{2}$ and recall $\Psi$ from \eqref{psidefinition}. As $\lambda + d$ tends to $\infty$, we have that
    \begin{flalign}
    \label{pxexponential1intro}
    & \displaystyle\inf_{p \in \mathcal{P}_d} \displaystyle\sup_{x \in [0, B]} \big| p(x) - e^{-x} \big| = (\lambda + d)^{O(1)} \exp (\Psi_{d, \lambda} - \lambda),  \\
     & \displaystyle\inf_{p \in \mathcal{P}_d} \displaystyle\sup_{x \in [0, B]} \big| p(x) - e^x \big| = (\lambda + d)^{O(1)} \exp (\Psi_{d, \lambda} + \lambda). \label{pxexponential2intro}
    \end{flalign}
    \end{corollary}

    Now \Cref{estimatedegree2} and \Cref{estimatedegree1} will follow from an explicit analysis of \eqref{pxexponential1intro} and \eqref{pxexponential2intro}, respectively. For the purposes of this outline, we will omit the remaining details of analyzing the asymptotics of these expressions (referring to \Cref{ProofEstimate} for a more detailed exposition). However, let us briefly explain how the transitions from $B = \omega \big( \log (\delta^{-1}) \big)$ to $B = o \big( \log (\delta^{-1}) \big)$ arise, for example in \Cref{estimatedegree2}. 
    
    They ultimately follow from the fact that the function $\Psi_{v, \lambda} - \lambda$ behaves differently depending on whether $v = O(\lambda)$ or $v = \Omega (\lambda)$, as different terms in the definition of $\Psi_{v, \lambda}$ are dominant in each of these settings; this is stated more precisely through the following lemma, which will be established in \Cref{EstimatesPsi} below.

	\begin{lemma}
	
	\label{psivlambdakappa} 
	
	Let $\lambda \ge \frac{1}{2}$ be a real number, $v \ge 0$ be an integer, denote $\kappa = \frac{v}{\lambda}$, and recall the function $\Psi_{v, \lambda}$ from \eqref{psidefinition}. 
	
	\begin{enumerate} 
	
	\item If $v \le 2 \lambda$ (that is, $\kappa \le 2$) then $\Psi_{v, \lambda} - \lambda = - \displaystyle\frac{v^2}{2 \lambda} \big( 1 + O(\kappa) \big)$.
	
	\item If $v \ge 2 \lambda$ (that is, $\kappa \ge 2$) then $\Psi_{v, \lambda} - \lambda= - v \log \Big( \displaystyle\frac{v}{\lambda} \Big) \Big( 1 + O \big( (\log \kappa)^{-1} \big) \Big)$.
	
	\end{enumerate}

	\end{lemma} 

    \noindent In particular, the first part of \Cref{psivlambdakappa} gives rise to the first part of \Cref{estimatedegree2}, and the second part of \Cref{psivlambdakappa} gives rise to the third part of \Cref{estimatedegree2}.

    Finally, we need one last tool to prove the fourth part of \Cref{estimatedegree2}. In \Cref{estimatedegree1} as well as the first three parts of \Cref{estimatedegree2}, our degree lower bound ultimately followed by finding a large coefficient in the Chebyshev expansion of $e^{\lambda x + \lambda}$ or $e^{\lambda x - \lambda}$ and then applying \Cref{ajcj}. However, when $B$ (and hence $\lambda$) is very large compared to the desired error, the Chebyshev expansion of $e^{\lambda x - \lambda}$ actually has no sufficiently large coefficients.
    
    We instead take a different approach in this last case. We observe that any polynomial $p$ satisfying the bound $\sup_{x \in [0, B]} \big| p(x) - e^{-x} \big| < \delta$ must have,
    \begin{itemize}
        \item $|p(x)| \leq 2\delta$ in the entire interval $x \in [\log(\delta^{-1}), B]$, and
        \item $p(0) \geq 1-\delta$.
    \end{itemize}
    It is known (see \Cref{extremalcheby} below) that the polynomial of lowest degree achieving these two properties must in fact be a (rescaled) Chebyshev polynomial. We then show that a Chebyshev polynomial requires degree $\Omega(\sqrt{B \log (\delta^{-1})})$ to realize these properties, from which the desired result follows.

\section{Preliminaries} \label{sec:prelims}

\subsection{Notation}

For a nonnegative integer $d$, we write $\mathcal{P}_d$ to denote the set of polynomials $p : \R \to \R$ with real coefficients of degree at most $d$.
For a Boolean predicate $P$, we write $$\textbf{1}_{P} = \begin{cases} 1 &\mbox{if } P \mbox{ is true,} \\
0 &\mbox{if } P \mbox{ is false.} \end{cases}$$ 
All logarithms in this paper are assumed to have base $e$, and we similarly write $\exp(x) := e^x$.

\subsection{Chebyshev Polynomials}

In this paper we make heavy use of the Chebyshev polynomials. Chebyshev polynomials appear prominently throughout polynomial approximation theory, and have been used in numerous other areas of theoretical computer science, including in Boolean function analysis and quantum computing; see e.g.,~\cite{bun2021guest}. Here we define them and give some of their well-known properties which will be important in our proofs. We refer the reader to~\cite{P} for more details.

The degree $d$ \emph{monic Chebyshev polynomial} $Q_d (x) \in \mathcal{P}_d$ is defined in many equivalent ways:
\begin{enumerate}[{Definition} 1]
    \item Set $Q_0 (x) = 1$ and, for each $d \ge 1$, define $Q_d (x)$ by imposing that, for each $\theta \in [0, 2 \pi]$, we have $Q_d (\cos \theta) = 2^{1 - d} \cos (d \theta)$. 
    \item $Q_0(x) = 1$, $Q_1(x) = x$, and for $d \geq 2$ we have $Q_d(x) = x \cdot Q_{d-1}(x) - \frac12 Q_{d-2}(x)$.
    \item $Q_0 (x) = 1$ and for each $d \ge 1$ we have $Q_d(x) = 2^{1-d} \cdot \sum_{k=0}^{\lfloor d/2 \rfloor} \binom{d}{2k} (x^2-1)^k x^{d-2k}$.
\end{enumerate}

In particular, $Q_d(x)$ is an even function when $d$ is even, and an odd function when $d$ is odd. From Definition 1, one observes several simple properties of $Q_d(x)$ for all $d>0$ for $x \in [-1,1]$:
\begin{itemize}
    \item For all $x \in [-1,1]$, we have $2^{d-1} \cdot Q_d(x) \in [-1,1]$.
    \item All $d$ roots of $Q_d(x)$ lie in $[-1,1]$, and they lie at the points $x = \cos \big( \pi (2k+1)/2d  \big)$ for each integer $0 \leq k < d$.
    \item On the interval $[-1,1]$, the extrema of $Q_d$ are located at the points $x = \cos\left( \pi k / d \right)$ for each integer $0 \leq k \leq d$. $Q_d(x)$ alternates between the values $2^{1-d}$ and $-2^{1-d}$ at these extrema, starting at $Q_d(\cos(0)) = Q_d(1) = 2^{1-d}$.
\end{itemize}  

Outside of the interval $[-1,1]$, it is well-known that the Chebyshev polynomials exhibit useful extremal properties.

\begin{fact} \label{extremalcheby}
For every integer $d>0$, every polynomial $p(x) \in \mathcal{P}_d$ of degree $d$ such that $2^{d-1} \cdot p(x) \in [-1,1]$ for all $x \in [-1,1]$, and every real $x' \notin [-1,1]$, we have $|Q_d(x')| \geq |p(x')|$.
\end{fact}

\begin{proof}
Assume to the contrary that there is an $x' \notin [-1,1]$ such that $|Q_d(x')| < |p(x')|$. By rescaling $p$ by a factor in the range $(|Q_d(x')|/|p(x')|,1)$, we can further assume that $2^{d-1} \cdot p(x) \in (-1,1)$ for all $x \in [-1,1]$. By the symmetry of $Q_d(x)$ (and by negating $p$ if necessary), we may also assume without loss of generality that $x'>1$ and that $p(x') > Q_d(x') > 0$ are positive.

Define the difference polynomial $g(x) = p(x) - Q_d(x)$, which has degree at most $d$. Consider the $d+2$ points $x_{-1} > x_0 > x_1 > x_2 > \cdots > x_{d} \in \R$ given by \begin{itemize}
    \item $x_{-1} = x'$, and
    \item $x_k = \cos\left( \pi k / d \right)$ for each integer $0 \leq k \leq d$.
\end{itemize}
We have $g(x_{-1}) = p(x') - Q_d(x') > 0$ by assumption. For even $k \geq 0$ we have $Q_d(x_k) = 2^{1-d}$ and $|p(x_k)| < 2^{1-d}$, and so $g(x_k) < 0$. Similarly, for odd $k>0$ we have $g(x_k) > 0$. Hence, $g(x)$ alternates signs at least $d+2$ times in the interval $[x_d, x_{-1}]$, meaning it has at least $d+1$ roots in that interval, a contradiction.
\end{proof}

In fact, $Q_d(1+\eps)$ for small $\eps>0$ is very closely approximated by an exponential in $\sqrt{\eps}$:

\begin{fact} \label{chebyasympt}
For any $\eps > 0$ we have $Q_d(1+\eps) = 2^{-d} \cdot e^{d \sqrt{2 \eps} (1+O(\sqrt{\eps}))}$. 
\end{fact}

\begin{proof}
Extending Definition 1 of $Q_d(x)$ to $x \notin [-1,1]$, we find that
$$Q_d(x) = 2^{-d} \left( \left( x - \sqrt{x^2 - 1} \right)^d +  \left( x + \sqrt{x^2 - 1} \right)^d  \right), \qquad \text{for each $x$ with $|x|>1$.}$$

\noindent For $x = 1+\eps$ with $\eps>0$, we thus get
$$Q_d(1+\eps) = 2^{-d} \left( \left( 1 - \sqrt{2\eps} + O(\eps) \right)^d +  \left( 1+\sqrt{2\eps } + O(\eps) \right)^d  \right) = 2^{-d} e^{d \sqrt{2 \eps} (1+O(\sqrt{\eps}))},$$
as desired.
\end{proof}

	\subsection{Chebyshev Expansion Coefficients}
	
	\label{ProofSumaj}
	
	In this section we prove \Cref{ajcj}, which shows how the coefficients of a function $f$ written in the basis of Chebyshev polynomials can be used to bound how well $f$ can be approximated by low-degree polynomials.

	\begin{proof}[Proof of \Cref{ajcj}] 
	
	    Observe for any $d \in \mathbb{Z}_{\ge 0}$ and $x \in [-1, 1]$ that $\big| Q_d (x) \big| \le 2^{1 - d}$, which follows from \eqref{qdtheta} after setting $x = \cos \theta$. This implies the absolute convergence of $f(x)$ for $x \in [-1, 1]$, since $\sum_{j = 0}^{\infty} |a_j| < \infty$. So, it remains to establish \eqref{pxfx}. 
		
		To establish the upper bound there, define
		\begin{flalign}
		\label{hdx}
		H_D (x) = \displaystyle\sum_{j = 0}^{D - 1} 2^{j - 1} a_j Q_j (x).
		\end{flalign}
		
		\noindent Since $\big| Q_d (x) \big| \le 2^{1 - d}$ for each $x \in [-1, 1]$, we have that 
		\begin{flalign}
		\label{hdxfxestimate} 
		\displaystyle\sup_{x \in [-1, 1]} \big| H_D (x) - f (x) \big|  = \displaystyle\sup_{x \in [-1, 1]} \Bigg| \displaystyle\sum_{j = D}^{\infty} 2^{j - 1} a_j Q_j (x)  \Bigg| \le \displaystyle\sum_{j = D}^{\infty} |a_j|,
		\end{flalign}
		
		\noindent which proves the upper bound in \eqref{pxfx}.
		
		To establish the lower bound, fix $p \in \mathcal{P}_{D - 1}$, and let $c_0, c_1, \ldots \in \mathbb{R}$ satisfy 
		\begin{flalign*}
		p(x) = \displaystyle\sum_{j = 0}^{D - 1} 2^{j - 1} c_j Q_j (x), \quad \text{and $c_j = 0$ for $j \ge D$}. 
		\end{flalign*}
		
		\noindent Then, applying \eqref{qdtheta}, we obtain 
		\begin{flalign}
		\label{pxfx1} 
		\begin{aligned}
		p(\cos \theta) - f(\cos \theta) & = \displaystyle\sum_{j = 0}^{\infty} 2^{j - 1} (a_j - c_j) Q_j (\cos \theta)  = \displaystyle\frac{a_0 - c_0}{2} + \displaystyle\sum_{j = 1}^{\infty} (a_j - c_j) \cos (j \theta).
		\end{aligned}
		\end{flalign} 
		
		\noindent Define $b_0, b_1, \ldots \in \mathbb{R}$ by setting $b_0 = \frac{a_0 - c_0}{2}$ and  $b_j = a_j - c_j$ for $j > 0$, and observe that 
		\begin{flalign*} 
		\int_0^{2 \pi} \cos (j \theta) \cos (k \theta) d \theta = 0, \qquad \text{for $j \ne k$},
		\end{flalign*} 
		
		 \noindent and 
		 \begin{flalign*} 
		 \displaystyle\int_0^{2\pi} \cos^2 (k \theta) d \theta = \displaystyle\frac{1}{|k|} \displaystyle\int_0^{2 \pi} \cos^2 ( \theta) d \theta = \pi |k|^{-1} \textbf{1}_{k \ne 0} + (2 \pi) \textbf{1}_{k = 0},
		\end{flalign*}
		
		\noindent it follows that
		\begin{flalign*}
		\displaystyle\int_0^{2 \pi} \big| p (\cos \theta) - f (\cos \theta) \big|^2 & = \displaystyle\int_0^{2 \pi} \Bigg( \displaystyle\sum_{j = 0}^{\infty} b_j \cos (j \theta) \Bigg)^2 d \theta \\ 
		& = \displaystyle\sum_{j = 0}^{\infty} \displaystyle\sum_{k = 0}^{\infty} \displaystyle\int_0^{2 \pi} b_j b_k \cos (j \theta) \cos (k \theta) d \theta = 2 \pi b_0^2 + \pi \displaystyle\sum_{k = 1}^{\infty} \displaystyle\frac{b_k^2}{k}.
		\end{flalign*}
		
		\noindent Thus, since $b_j = a_j$ for $j \ge D$ (since $c_j = 0$ for $j \ge D$), we deduce
		\begin{flalign*}
		\displaystyle\sup_{x \in [-1, 1]} \big| p (x) - f(x) \big|^2 = 
		\displaystyle\sup_{\theta \in [0, 2 \pi]} \big| p (\cos \theta) - f(\cos \theta) \big|^2 & \ge \displaystyle\frac{1}{2 \pi} \displaystyle\int_0^{2 \pi} \big| p (\cos \theta) - f (\cos \theta) \big|^2 d \theta \\
		& \ge \displaystyle\frac{1}{2} \displaystyle\sum_{k = D}^{\infty} \displaystyle\frac{a_k^2}{k},
		\end{flalign*}
		
		\noindent which yields the proposition. 
	\end{proof}

    Finally, we recall a well-known \cite{P} identity expressing a monomial as an explicit linear combination of Chebyshev polynomials. 
    
    \begin{lemma}[{\cite[(2.14)]{P}}]
    
    \label{xqsum}
    
    For any integer $n \ge 0$, we have that 
    \begin{flalign*}
	x^n = \displaystyle\sum_{k = 0}^{\lfloor n / 2 \rfloor} 2^{- 2k} \binom{n}{k} Q_{n - 2k} (x).
	\end{flalign*}
    
    \end{lemma}

\section{Degree Bounds for Polynomial Approximations}

\label{ProofEstimate}

	\subsection{Estimating $A_{v, \lambda}$ and $B_{v, \lambda}$}

	\label{Proofavlambdabvlambda}

	In this section we analyze $A_{v, \lambda}$ and $B_{v, \lambda}$ from \eqref{avlambdadefinition}. We begin with the proof of \Cref{amlambda}.
	
	\begin{proof}[Proof of \Cref{amlambda}] 
	
	We only establish the first statement in \eqref{absum}, as the proof of the second is entirely analogous. To that end, first using the series expansion for $e^{-z} = \sum_{n = 0}^{\infty} \frac{(-z)^n}{n!}$ and then applying \Cref{xqsum}, yields
	\begin{flalign*}
	e^{-\lambda x - \lambda} = e^{-\lambda} \displaystyle\sum_{n = 0}^{\infty} \displaystyle\frac{(-\lambda)^n}{n!} x^n =  e^{-\lambda} \displaystyle\sum_{n = 0}^{\infty} \displaystyle\frac{(-\lambda)^n}{n!} \displaystyle\sum_{k = 0}^{\lfloor n / 2 \rfloor} 2^{- 2k} \binom{n}{k} Q_{n - 2k} (x).
	\end{flalign*}
	
	\noindent Then, by setting $v = n - 2k$, we obtain
	\begin{flalign*}
	e^{-\lambda x - \lambda} = e^{-\lambda} \displaystyle\sum_{v = 0}^\infty  Q_v (x) \displaystyle\sum_{k = 0}^{\infty}  \displaystyle\frac{(-\lambda)^{v + 2k}}{(v + 2k)!} 2^{-2k} \binom{v + 2k}{k} = \displaystyle\sum_{v = 0}^\infty 2^{v - 1} A_{v, \lambda} Q_v (x),
	\end{flalign*} 
	
	\noindent from which we deduce the lemma. 
	\end{proof}

	We next have the following proposition that more precisely formulates \Cref{avlambdalimit}. 
	
		\begin{proposition}
		
	\label{avlambdalimit2}
	
	There exist constants $C, c > 0$ such that the following holds. For any integer $v \ge 0$ and real number $\lambda \ge \frac{1}{2}$, recall the quantities $A_{v, \lambda}$ and $B_{v, \lambda}$ from \eqref{avlambdadefinition} and $\Psi_{\lambda, v}$ from \eqref{psidefinition}. 
	
	\begin{enumerate} 
	
	\item For any $v$ and $\lambda$ as above, we have that 
		\begin{flalign*}
		& c (v + \lambda)^{-1} \exp \big( \Psi_{v, \lambda} - \lambda \big) \le (-1)^v A_{v, \lambda} \le C (v + \lambda) \exp \big( \Psi_{v, \lambda} - \lambda \big); \\
		& c (v + \lambda)^{-1} \exp \big( \Psi_{v, \lambda} + \lambda \big) \le B_{v, \lambda} \le C (v + \lambda) \exp \big( \Psi_{v, \lambda} + \lambda \big).
		\end{flalign*}
		
	\item If $v \le \lambda$, then 
	\begin{flalign*} 
	c (v + \lambda)^{-1} \exp \big( \Psi_{v, \lambda} - \lambda \big) \le (-1)^v A_{v, \lambda} \le C \lambda^{-1/2} \exp \big( \Psi_{v, \lambda} - \lambda \big).
	\end{flalign*}
	\end{enumerate} 
		
	\end{proposition}
	
	To to establish the above result, observe that the two quantities $A_{v, \lambda}$ and $B_{v, \lambda}$ are quite similar, in that they both involve certain sum given by
	\begin{flalign}
	\label{evlambdadefinition}
	E_{v, \lambda} = \sum_{n - v \in 2 \mathbb{Z}_{\ge 0}} \displaystyle\frac{\lambda^n}{2^n n!} \binom{n}{\frac{n - v}{2}}.a
	\end{flalign}
    
    \noindent We therefore require the following proposition that estimates $E_{v, \lambda}$. Observe its second statement slightly improves the upper bound in its first statement for $\lambda$ large (which will be useful in analyzing $d_{B; \delta}$ in the regime of large $B$ below).
    
    \begin{proposition}
    
    \label{estimateevlambda} 
    
    There exist constants $C, c > 0$ such that the following holds. For any integer $v \ge 0$ and real number $\lambda \ge \frac{1}{2}$, recall the quantities $E_{v, \lambda}$ and $\Psi_{v, \lambda}$ from \eqref{evlambdadefinition} and \eqref{psidefinition}, respectively. 
    
    \begin{enumerate}
    	\item If $v \ge \lambda$, then $c (v + \lambda)^{-1} \exp \big( \Psi_{v, \lambda} \big) \le E_{v, \lambda} \le C (v + \lambda) \exp \big( \Psi_{v, \lambda} \big)$.
		
	\item If $v \le \lambda$, then $c (v + \lambda)^{-1} \exp \big( \Psi_{v, \lambda} \big) \le E_{v, \lambda} \le C \lambda^{-1/2} \exp \big( \Psi_{v, \lambda} \big)$.
		
	\end{enumerate} 
	\end{proposition}
	
	Proposition \ref{avlambdalimit} now follows directly from \Cref{estimateevlambda}. 

	\begin{proof}[Proof of \Cref{avlambdalimit2} Assuming \Cref{estimateevlambda}]

    Given \Cref{estimateevlambda}, \Cref{avlambdalimit2} follows from the facts that $(-1)^v A_{v, \lambda} = 2 e^{-\lambda} E_{v, \lambda}$ and $B_{v, \lambda} = 2 e^{\lambda} E_{v, \lambda}$.
    \end{proof}
	
	We must thus establish \Cref{estimateevlambda}, to which end we begin with the following lemma. 
	
	\begin{lemma} 
	
	\label{estimateevlambda1} 
	
	For any integer $v \ge 1$ and real number $\lambda \ge \frac{1}{2}$, we have that
	\begin{flalign}
	\label{evlambdafestimate} 
	E_{v, \lambda} = \Theta \Bigg( \displaystyle\sum_{n - v \in 2 \mathbb{Z}_{\ge 0}} \big( n^2 - v^2 + n \big)^{-1 / 2} \exp \big( F (n) \big) \Bigg),
	\end{flalign}
	
	\noindent where for any real number $n \ge v$ we have defined $F (n) = F (n)$ by 
	\begin{flalign}
	\label{f}
	F (n) = n \log \lambda - n \log 2 + n - \Big( \frac{n - v}{2} \Big) \log \Big( \frac{n - v}{2} \Big) - \Big( \frac{n + v}{2} \Big) \log \Big( \frac{n + v}{2} \Big).
	\end{flalign}
	
	\end{lemma} 
	
	\begin{proof} 
	
	The explicit form \eqref{evlambdadefinition} for $E_{v, \lambda}$ and the Stirling estimate $n! = \Theta \big( (n + 1)^{n + 1 /2} e^{- n} \big)$, which holds uniformly in $n \ge 0$, together imply 
	\begin{flalign*}
	E_{v, \lambda} = \Theta \Bigg( \displaystyle\sum_{n - v \in 2 \mathbb{Z}_{\ge 0}} \big( (n + 2)^2  - v^2 \big)^{-1 / 2} \exp \big( F (n) \big) \Bigg).
	\end{flalign*} 
	
	\noindent From this, we deduce the lemma since $(n + 2)^2 - v^2 = \Theta (n^2 - v^2 + n)$, uniformly in $n \ge v$. 
	\end{proof} 
	
	The right side of \eqref{evlambdafestimate} will be dominated by the terms near which $F$ is maximized, so we next perform a critical point analysis on $F$. 
	
	\begin{lemma} 
	
	\label{fn0}
	
	Fix an integer $v \ge 1$ and a real number $\lambda \ge \frac{1}{2}$, and set $n_0 = n_0 (v, \lambda) =  \sqrt{v^2 + \lambda^2}$. There exist constants $c, C > 0$ (independent of $v$ and $\lambda$) such that the following holds.
	
	\begin{enumerate}
	    \item The function $F (n)$ is maximized at $n = n_0$, and $F (n_0) = \Psi_{v, \lambda}$ (recall \eqref{psidefinition}). 
	   
	    \item For $z \ge 2 \lambda$, we have that $F (n_0 + z) \le F (n_0) - c z$.
	    
	    \item For at least one choice of $m \in \big\{ \lfloor n_0 \rfloor, \lceil n_0 \rceil \big\}$, we have that $F (m) \ge F (n_0) - C$.
	    
	    \item If $v \le \lambda$, then for any $z \in [v - n_0, 2 \lambda]$ we have that $F (n_0 + z) \le F (n_0) - c \lambda^{-1} z^2$.
	     
	\end{enumerate}
	\end{lemma} 
	
	\begin{proof}
	
	Using the explicit form \eqref{f} for $F (n)$, we deduce for $n \ge v$ that 
	\begin{flalign}
	\label{fnderivatives}
	& F' (n) = \log \lambda - \displaystyle\frac{1}{2} \log (n^2 - v^2), \quad \text{and} \quad F'' (n) = \displaystyle\frac{n}{v^2 - n^2} \in \bigg[ - \displaystyle\frac{1}{n}, 0 \bigg],
	\end{flalign}
	
	\noindent which implies that $F' (n_0) = 0$ and that $F$ is maximized at $n_0$. Upon insertion into \eqref{f} (and recalling \eqref{psidefinition}), we also find that 
	\begin{flalign*}
	F (n_0) = \sqrt{v^2 + \lambda^2} + v \log \Big( \displaystyle\frac{\sqrt{v^2 + \lambda^2} - v}{\lambda} \Big) = \Psi_{v, \lambda},
	\end{flalign*} 
	
	\noindent which verifies the first statement of the lemma. 
	
	To establish the second, first observe since $F'' (n) \le 0$ and $n_0^2 - v^2 = \lambda^2$ that
	\begin{flalign}
	\label{fnlambda}
	F' (n_0 + \lambda) = \log \lambda - \displaystyle\frac{1}{2} \log \big( (n_0 + \lambda)^2 - v^2 \big) \le \log \lambda - \displaystyle\frac{1}{2} \log (n_0^2 + \lambda^2 - v^2) = - \displaystyle\frac{\log 2}{2}.
	\end{flalign}
	
    \noindent Thus, we deduce for $z \ge 2 \lambda$ that
    \begin{flalign*}
    F (n_0 + z) - F(n_0) \le F (n_0 + z) - F (n_0 + \lambda) \le (z - \lambda) F' (n_0 + \lambda) \le \displaystyle\frac{\log 2}{2} (\lambda - z) \le - \displaystyle\frac{z \log 2}{4},
    \end{flalign*}
    
    \noindent where in the first inequality we used the fact that $F$ is maximized at $n_0$; in the second we used the fact that $F'' (n) \le 0$; in the third we used \eqref{fnlambda}; and in the fourth we used the fact that $z - \lambda \ge \frac{z}{2}$. This verifies the second statement of the lemma. 
	
	To show the third part of the lemma, we separately consider the cases when $v \le \lambda^2$ and $v \ge \lambda^2$. In the former situation $v \le \lambda^2$, we select $m = \lceil n_0 \rceil$. Since $F''' (n) = (n^2 + v^2) (n^2 - v^2)^{-2} \ge 0$, we then have for for $n \in [n_0, m]$ that
	\begin{flalign*}
	F' (n) \ge (n - n_0)^2 F'' (n_0) \ge - \displaystyle\frac{(n-n_0)^2}{n_0} \ge -\displaystyle\frac{1}{n_0} \ge \displaystyle\frac{n_0}{\lambda^2},
	\end{flalign*}
	
	\noindent using the last second identity in \eqref{fnderivatives}. Since $n_0 \le \lambda + v \le 3 \lambda^2$ (due to the facts that $\lambda \ge \frac{1}{2}$ and $v \le \lambda^2$), it follows that $F' (n) \ge -\frac{1}{3}$ for $n \in [n_0, m]$, which implies that $F (m) \ge F(n_0) - \frac{1}{3}$ if $v \le \lambda^2$. 
	
	Now instead suppose that $v \ge \lambda^2$, in which case we take $m = \lfloor n_0 \rfloor = v$. Then, using the second identity in \eqref{fnderivatives}, we obtain
	\begin{flalign*}
	F (n_0) - F (v) = \displaystyle\int_0^{n_0 - v} F' (v + z) dz = - \displaystyle\frac{1}{2} \displaystyle\int_0^{n_0 - v} \log \bigg( \displaystyle\frac{2vz + z^2}{\lambda^2} \bigg) dz \le  - \displaystyle\frac{1}{2} \displaystyle\int_0^{n_0 - v} \log \bigg( \displaystyle\frac{vz}{\lambda^2} \bigg) dz.
	\end{flalign*}
	
	\noindent Now set $r = \frac{\lambda^2}{v}$, and observe that $n_0 - v \le r \le 1$. This yields 
	\begin{flalign*}
	F(m) \ge F(n_0) + \displaystyle\frac{1}{2} \displaystyle\int_0^r \log \bigg( \displaystyle\frac{z}{r} \bigg) dz = F (n_0) +  \displaystyle\frac{r}{2} \displaystyle\int_0^1 (\log y) dy = F (n_0) - \displaystyle\frac{r}{2} \ge F (n_0) - 1,
	\end{flalign*} 
	
	\noindent where in the first equality we changed variables $z = ry$. This verifies the third part of the lemma. 
	
	To establish the fourth part of the lemma, let us first consider the case when $z \in [v - n_0, 0]$. Then, since $F' (n_0) = 0$ and $F''' (n) = (n^2 + v^2) (n^2 - v^2)^{-2} \ge 0$ for each $n \ge v$, we have that 
	\begin{flalign*}
	F(n_0 + z) \le F (n_0) + \displaystyle\frac{z^2 F'' (n_0)}{2} = F (n_0) - \displaystyle\frac{z^2 (\lambda^2 + v^2)^{1/2}}{2 \lambda^2}, \quad \text{for $z \in [v - n_0, 0]$},
	\end{flalign*} 
	
	\noindent where in the last equality we used the second identity in \eqref{fnderivatives} for $F' (n)$ (and the fact that $n_0 = \sqrt{\lambda^2 + v^2}$). Thus, we deduce that
	\begin{flalign}
	\label{fn0zn01}
	F (n_0 + z) \le F (n_0) - \displaystyle\frac{z^2}{2 \lambda}, \quad \text{for $z \in [v - n_0, 0]$}.
	\end{flalign}
	
	Next we consider the case when $z \in [0, 2 \lambda]$. In this case, the second identity in \eqref{fnderivatives} implies for each $m \in [n_0, n_0 + 2 \lambda]$ that 
	\begin{flalign*}
	F'' (m) = \displaystyle\frac{m}{v^2 - m^2} \le -\displaystyle\frac{\lambda}{m^2} \le -\displaystyle\frac{1}{16 \lambda},
	\end{flalign*}
	
	\noindent where in the last inequality we used the fact that $m \le n_0 + 2 \lambda \le 4 \lambda$ (as $v \le \lambda$ and $n_0 = \sqrt{v^2 + \lambda^2}$). Thus, it follows from the fact that $F' (n_0) = 0$ that 
	\begin{flalign}
	\label{fn0zn02}
	F (n_0 + z) \le F (n_0) - \displaystyle\frac{z^2}{32 \lambda}, \quad \text{for each $z \in [0, 2 \lambda]$}.
	\end{flalign}
	
	\noindent Now the fourth statement of the lemma follows from \eqref{fn0zn01} and \eqref{fn0zn02}. 
	\end{proof}
	
	Now we can establish \Cref{estimateevlambda}. 
	
	\begin{proof}[Proof of \Cref{estimateevlambda}]

    We begin by establishing the lower bound on $E_{v, \lambda}$, simultaneously in both cases $v \ge \lambda$ and $v \le \lambda$. To that end, we first apply \Cref{estimateevlambda1} and then use the third part of \Cref{fn0} to bound the sum on the right side of \eqref{evlambdafestimate} its summand corresponding to a suitable choice of index $m \in \big\{ \lfloor n_0 \rfloor, \lceil n_0 \rceil \big\}$. This yields constants $c_1, c_2 > 0$ such that
	\begin{flalign*}
	E_{v, \lambda} \ge c_1 \big( \lambda^2 + v^2 \big)^{-1 / 2} \exp \big( - F (m) \big) \ge c_1 (\lambda + v)^{-1} \exp \big( F (n_0) - c_2 \big),
	\end{flalign*}
	
	\noindent which implies the lower bound on $E_{v, \lambda}$ (in either case $v \ge \lambda$ or $v \le \lambda$), since $F (n_0) = \Psi_{\lambda, v}$.
	
	To establish the upper bound in the case $v \ge \lambda$, observe that \Cref{estimateevlambda1}; the fact that $F(n) \le F(n_0) = \Psi_{v, \lambda}$ for each $n \in [v, n_0 + 2 \lambda]$ (by the first part of \Cref{fn0}); and the existence of a constant $c > 0$ such that $F(n_0 + z) \le F(n_0) - cz$ for $z \ge 2 \lambda$ (by the second part of \Cref{fn0}) together yield a constant $C_1 > 0$ such that 
	\begin{flalign*}
	E_{v, \lambda} \le  C_1 (n_0 + 2 \lambda) \exp (\Psi_{v, \lambda}) \displaystyle\sum_{z = 2 \lambda}^{\infty} e^{-cz} \le 2 c^{-1} C_1 (n_0 + 2 \lambda) \exp (\Psi_{v, \lambda}).
	\end{flalign*}
	
	\noindent This establishes the upper bound on $E_{v, \lambda}$ when $v \ge \lambda$. 

    In the latter case $v \le \lambda$, we proceed as above but additionally use the facts that $F(n_0 + z) \le F(n_0) - c \lambda^{-1} z^2$ for $z \in [v - n_0, 2 \lambda]$ (by the fourth part of \Cref{fn0}) to deduce for some constant $C_2 > 0$ that
    \begin{flalign*}
	E_{v, \lambda} & \le  C_1 \exp (\Psi_{v, \lambda}) \Bigg( \displaystyle\sum_{z = v - n_0}^{2 \lambda} \big( (n_0 + z)^2 - v^2 + n_0 + z \big)^{-1 / 2} e^{-cz^2 / \lambda} + \displaystyle\sum_{z = 2 \lambda}^{\infty} e^{-cz} \Bigg) \\
	& \le C_1 \exp (\Psi_{v, \lambda}) \Bigg( \displaystyle\sum_{|z| \le \lambda/4} \big( (n_0 + z)^2 - v^2 + n_0 + z \big)^{-1/2} e^{-cz^2 / \lambda} + \displaystyle\sum_{|z| \ge \lambda/4} e^{-cz^2 / \lambda} + \displaystyle\sum_{z = 2 \lambda}^{\infty} e^{-cz} \Bigg) \\ 
	& \le  C_1 \exp (\Psi_{v, \lambda}) \Bigg( 12 \lambda^{-1} \displaystyle\sum_{|z| \le \lambda/4} e^{-cz^2 / \lambda} + 35 c^{-1} e^{-2 \lambda} \Bigg) \le C_2 \lambda^{-1 / 2} \exp (\Psi_{v, \lambda}),
	\end{flalign*}
	
	\noindent where in the third inequality we used the fact that $(n_0 + z)^2 - v^2 \ge \frac{\lambda^2}{144}$ for $|z| \le \frac{\lambda}{4}$ (as $n_0 = \sqrt{\lambda^2 + v^2} \ge \frac{\lambda}{3} + v$ for $\lambda \ge v$). This establishes the upper bound on $E_{v, \lambda}$ when $v \le \lambda$.  
	\end{proof}

    \subsection{Estimates for the Minimum Polynomial Approximation Error}

    \label{EstimatesPsi}
    
    In this section we establish the following proposition, which is the variant of \Cref{pexponential} that will be useful for our purposes.
    	
	 \begin{proposition}
    
    \label{pexponential2}
    
    There exist constants $C, c > 0$ such that the following holds. Let $d \ge 1$ be an integer, and let $B \ge 1$ be a real number. Set $\lambda = \frac{B}{2}$ and recall $\Psi$ from \eqref{psidefinition}. 
    
    \begin{enumerate} 
    
    \item For any $B$ and $d$ as above, we have
    \begin{flalign}
    \label{exponential1p}
     & c (d + \lambda)^{-3/2} \exp (\Psi_{d, \lambda} - \lambda) \le \displaystyle\inf_{p \in \mathcal{P}_d} \displaystyle\sup_{x \in [0, B]} \big| p(x) - e^{-x} \big| \le C (d + \lambda)^2 \exp (\Psi_{d, \lambda} - \lambda); \\
     & c (d + \lambda)^{-3/2} \exp (\Psi_{d, \lambda} + \lambda) \le \displaystyle\inf_{p \in \mathcal{P}_d} \displaystyle\sup_{x \in [0, B]} \big| p(x) - e^x \big| \le C (d + \lambda)^2 \exp (\Psi_{d, \lambda} + \lambda). \label{exponential2p}  
    \end{flalign}
    
    \item If $B \ge 2d$, then we have that
    \begin{flalign}
    \label{exponential3p}
    & c (d + \lambda)^{-3/2} \exp (\Psi_{d, \lambda} - \lambda) \le \displaystyle\inf_{p \in \mathcal{P}_d} \displaystyle\sup_{x \in [0, B]} \big| p(x) - e^{-x} \big| \le C \exp (\Psi_{d, \lambda} - \lambda).
    \end{flalign}
    
    \end{enumerate}
    
    \end{proposition}
    
    We will establish \Cref{pexponential2} as a consequence of \Cref{ajcj} and \Cref{avlambdalimit2}. However, before doing so, it will be useful to obtain some properties for $\Psi_{v, \lambda}$. Therefore, we first prove \Cref{psivlambdakappa}.

	\begin{proof}[Proof of \Cref{psivlambdakappa}]
	
	First observe that
	\begin{flalign*}
	\Psi_{v, \lambda} - \lambda = \lambda \Big( \sqrt{\kappa^2 + 1} - 1 + \kappa \log \big( \sqrt{\kappa^2 + 1} - \kappa \big) \Big).
	\end{flalign*}
	
	\noindent In particular, if $\kappa \le 2$ then using the series expansions 
	\begin{flalign*}
	\sqrt{z^2 + 1} = 1 + \displaystyle\frac{z^2}{2} + O(z^3), \quad \text{and} \quad \log (1 + z) = z + O (z^2), \quad \text{valid for $z \in [0, 2]$},
	\end{flalign*}
	
	\noindent  we obtain that
	\begin{flalign*}
	\Psi_{v, \lambda}  = \lambda \bigg( \displaystyle\frac{\kappa^2}{2} + O(\kappa^3) + \kappa \log \Big( 1 - \kappa + O(\kappa^2) \Big) \bigg) = - \displaystyle\frac{\kappa^2 \lambda}{2} \big( 1 + O(\kappa) \big).
	\end{flalign*}
	
	\noindent If instead $\kappa \ge 2$ then using the series expansion
	 \begin{flalign*}
	 \sqrt{z^2 + 1} = z + \displaystyle\frac{1}{2 z} + O (z^{-2}), \quad \text{and} \quad \log (z^{-1} + z^{-2}) = - \log z - O(z^{-1}), \quad \text{valid for $|z| \ge 2$}, 
	 \end{flalign*}
	 
	 \noindent we deduce that
	\begin{flalign*}
	\Psi_{v, \lambda} = \lambda \Bigg( \kappa + O \Big( \displaystyle\frac{1}{\kappa} \Big) + \kappa \log \bigg( \displaystyle\frac{1}{2 \kappa} + O \Big( \displaystyle\frac{1}{\kappa^2} \Big)  \bigg) \Bigg) = - \lambda \kappa \log \kappa \Big( 1 + O \big( (\log \kappa)^{-1} \big) \Big).
	\end{flalign*}
	
	\noindent This establishes the lemma.
	\end{proof}

    Now we can establish \Cref{pexponential2}. 
    
    \begin{proof}[Proof of \Cref{pexponential2}]
    
    First observe that by rescaling (namely, replacing $x$ with $\lambda(x+1)$ or $-\lambda(x+1)$), we have that 
    \begin{flalign}
    \label{estimatepblambda} 
    \begin{aligned}
    & \displaystyle\inf_{p \in \mathcal{P}_d} \displaystyle\sup_{x \in [0, B]} \big| p(x) - e^x \big| = \displaystyle\inf_{p \in \mathcal{P}_d} \displaystyle\sup_{x \in [-1, 1]} \big| p(x) - e^{\lambda x + \lambda} \big|; \\
    & \displaystyle\inf_{p \in \mathcal{P}_d} \displaystyle\sup_{x \in [0, B]} \big| p(x) - e^{-x} \big| = \displaystyle\inf_{p \in \mathcal{P}_d} \displaystyle\sup_{x \in [-1, 1]} \big| p(x) - e^{-\lambda - \lambda x} \big|.
    \end{aligned}
    \end{flalign}
    
    \noindent Therefore, \Cref{ajcj} and the definitions of $A_{v, \lambda}$ and $B_{v, \lambda}$ from \eqref{avlambdadefinition}, together yield 
    \begin{flalign}
    \label{pestimateab}
    \begin{aligned}
    & \displaystyle\inf_{p \in \mathcal{P}_d} \displaystyle\sup_{x \in [0, B]} \big| p(x) - e^{-x} \big| \ge (2d)^{-1/2} \big| A_{v, \lambda} \big|; \qquad \displaystyle\inf_{p \in \mathcal{P}_d} \displaystyle\sup_{x \in [0, B]} \big| p(x) - e^x \big| \ge (2d)^{-1 / 2} \big| B_{v, \lambda} \big|. 
    \end{aligned}
    \end{flalign}
    
    \noindent Thus, the lower bounds on the minimal error for $e^{-x}$ and $e^x$ in both of the cases listed in the proposition follow from \eqref{pestimateab} and the lower bounds on $|A_{v, \lambda}|$ and $|B_{v, \lambda}|$ from the first part of \Cref{avlambdalimit2}.
    
    Now let us establish the upper bounds in this proposition; in what follows, $C > 0$ will denote a constant (uniform in $d$ and $\lambda)$ that might change between appearances. We first show \eqref{exponential2p}, to which end, observe that \eqref{estimatepblambda}, \Cref{ajcj}, and the upper bound for $B_{v, \lambda}$ from the first part of \Cref{avlambdalimit2} together yield 
    \begin{flalign}
    \label{pd}
     \displaystyle\inf_{p \in \mathcal{P}_d} \displaystyle\sup_{x \in [0, B]} \big| p(x) - e^{-x} \big| \le \displaystyle\sum_{v = d}^{\infty} B_{v, \lambda} \le C \displaystyle\sum_{v = d}^{\infty} (v + \lambda) \exp (\Psi_{v, \lambda} + \lambda).
    \end{flalign}
    
    \noindent Next, observe from the second part of \Cref{psivlambdakappa} that 
    \begin{flalign}
    \label{psivclambda}
    \Psi_{v, \lambda} + \lambda \le -v, \qquad \text{for $v > C \lambda$}.
    \end{flalign} 
    
    \noindent Further observe that 
    \begin{flalign}
    \label{psidecreasing}
    \text{$\Psi_{v, \lambda}$ is decreasing in $v \ge 0$ for fixed $\lambda$},
    \end{flalign} 
    
    \noindent since
    \begin{flalign}
    \label{psiderivative}
    \displaystyle\frac{\partial}{\partial v} \Psi_{v; \lambda} = \log \big( \sqrt{\kappa^2 + 1} - \kappa \big)  \le 0, \qquad \text{for $\kappa = \displaystyle\frac{v}{\lambda} \ge 0$}.
    \end{flalign} 
    
    From \eqref{pd}, \eqref{psivclambda}, and \eqref{psidecreasing}, it follows that 
    \begin{flalign*}
     \displaystyle\inf_{p \in \mathcal{P}_d} \displaystyle\sup_{x \in [0, B]} \big| p(x) - e^x \big| & \le C (d + \lambda) \exp \big( \Psi_{d, \lambda} + \lambda \big) \Bigg( d + C \lambda + \displaystyle\sum_{v \ge d + C \lambda} (v + \lambda) e^{-v} \Bigg) \\
     & \le C (d + \lambda)^2 \exp \big( \Psi_{d, \lambda} + \lambda \big).
    \end{flalign*}
    
    \noindent This establishes \eqref{exponential2p}; the proof of \eqref{exponential1p} is omitted as it is entirely analogous.
    
    Now let us establish the improved upper bound on the minimum error in the case when $B \ge 2d$. To that end, we as before apply \eqref{estimatepblambda}, \Cref{ajcj}, and the upper bound for $A_{v, \lambda}$ from the second part of \Cref{avlambdalimit2} to obtain
    \begin{flalign}
    \label{pdestimate1}
     \displaystyle\inf_{p \in \mathcal{P}_d} \displaystyle\sup_{x \in [0, B]} \big| p(x) - e^{-x} \big| \le C \displaystyle\sum_{v = d}^{\infty} \lambda^{-1/2} \exp (\Psi_{v, \lambda} - \lambda).
    \end{flalign}
    
    \noindent By \eqref{psiderivative}, we have that $\Psi_{v; \lambda} < 0$ for $v > 0$ and moreover that
    \begin{flalign*}
    \displaystyle\frac{\partial}{\partial v} \Psi_{v; \lambda} = \log \bigg( 1 - \displaystyle\frac{v}{\lambda} + O \Big( \displaystyle\frac{v^2}{\lambda^2} \Big) \bigg) = O \bigg( \displaystyle\frac{v^2}{\lambda^2} \bigg) - \displaystyle\frac{v}{\lambda}, \qquad \text{for $v \le 10 \lambda$}; \\
    \displaystyle\frac{\partial}{\partial v} \Psi_{v; \lambda} = \log \bigg( \displaystyle\frac{\lambda}{2v} + O \Big( \displaystyle\frac{\lambda^2}{v^2} \Big) \bigg) = \bigg( 1 + O \Big( \displaystyle\frac{\lambda}{v} \Big) \bigg) \log \bigg( \displaystyle\frac{\lambda}{2v} \bigg), \qquad \text{for $v \ge 10 \lambda$}.
    \end{flalign*} 
    
    \noindent In particular, there exist constants $c_1, c_2 > 0$ such that for $v, \lambda \ge \frac{1}{2}$ we have $\frac{\partial}{\partial v} \Psi_{v; \lambda} \le -c_1$ if $v \ge c_2 \lambda$ and $\frac{\partial}{\partial v} \Psi_{v; \lambda} \le - \frac{v}{2 \lambda}$ for $v \le c_2 \lambda$. Thus,
    \begin{flalign*}
    \Psi_{v; \lambda} - \Psi_{d; \lambda} \le -\displaystyle\frac{c_2}{4 \lambda} (v-d)^2, \quad \text{for $v \le c_2 \lambda$}; \qquad \Psi_{v; \lambda} - \Psi_{d; \lambda} \le C^{-1} (d - v), \quad \text{for $v \ge c_2 \lambda$},
    \end{flalign*}
    
    \noindent and so
    \begin{flalign}
    \label{pdestimate2}
    \displaystyle\sum_{v = d}^{\infty} \exp (\Psi_{v, \lambda} - \lambda) \le C \lambda^{1/2} \exp (\Psi_{d, \lambda} -\lambda).
    \end{flalign}
    
    \noindent The upper bound in \eqref{exponential3p} now follows from \eqref{pdestimate1} and \eqref{pdestimate2}.
    \end{proof}

	\subsection{Proofs of \Cref{estimatedegree2} and \Cref{estimatedegree1}}
	
	\label{ProofDegree}
	
    In this section we establish \Cref{estimatedegree2} and \Cref{estimatedegree1}. Recalling the function $G(x)$ and the quantity $\Psi_{v, \lambda}$ from \eqref{psidefinition} from these statements, both of these proofs will use the fact that 
    \begin{flalign}
    \label{psivlambdag}
    \Psi_{v, \lambda} = \lambda G \Big( \displaystyle\frac{v}{\lambda} \Big).
    \end{flalign} 

    \noindent We begin with the proof of \Cref{estimatedegree1}.

	\begin{proof}[Proof of \Cref{estimatedegree1}]
	
	Set $\lambda = \frac{B}{2}$. Observe that there exists a constant $c_1 > 0$ such that $G' (z) < -c_1$ whenever $z \in [z_* - c_1, z_* + c_1]$. Thus, $G(x) + 1 > c_1 (z-x_*)$, and so \eqref{psivlambdag} yields for some constant $c_2 > 0$ that
	\begin{flalign*} 
	(d + \lambda)^{-1} \exp (\Psi_{d, \lambda} + \lambda) > (d + \lambda)^{-1} \exp (c_2 \lambda^{1/2}) > 10,
	\end{flalign*} 
	
	\noindent if $d < z_* \lambda - \lambda^{1/2}$ and $\lambda$ is sufficiently large. Thus, by the lower bound in \Cref{pexponential2}, for any $\lambda$ sufficiently large and $\delta \le \frac{1}{2}$ we have 
	\begin{flalign}
	\label{db2}
	d_{B; \delta} \ge \big( z_* + o(1) \big) \lambda = \big( z_* + o(1) \big) \frac{B}{2}.
	\end{flalign}
	
	Now, assume first that $B = \omega \big( \log (\delta^{-1}) \big)$. Then, the upper bound in \Cref{pexponential2} implies that $d (B; \delta) \le d$ if $d$ satisfies $(d + \lambda)^2 \exp \big( \Psi_{d, \lambda} + \lambda \big) < \delta$. Since $G(z_*) = 0$ and $G' (z_*) < 0$, there exists a constant $C > 0$ such that 
	\begin{flalign*} 
	(d + \lambda)^2 \exp \big( \Psi_{d, \lambda} + \lambda \big) \le (d + \lambda)^2 e^{- C (d - z_* \lambda)}.
	\end{flalign*} 
	
	\noindent Since $\lambda = \frac{B}{2} = \omega \big( \log (\delta^{-1}) \big)$ implies that $(d + \lambda)^2 \exp \big( \Psi_{d, \lambda} + \lambda \big) < \delta$ for $d = \big( z_* + o(1) \big) \lambda = \big( z_* + o(1) \big) \frac{B}{2}$. Hence, in this case $d_{B; \delta} (e^x) \le \big( z_* + o(1) \big) \frac{B}{2}$, which by \eqref{db2} implies that $d_{B; \delta} (e^x) = \big( z_* + o(1) \big) \frac{B}{2}$.
	
	Now assume that $B = \big( 2r + o(1) \big) \log (\delta^{-1})$ for some fixed $r > 0$, so that $\lambda = \big( r + o(1) \big) \log (\delta^{-1})$. Suppose that $d < \mu' \lambda$ for some $\mu' < \mu(r)$. Then, $G (\mu') + 1 > -r^{-1}$ and so we have again using \eqref{psivlambdag} (and the fact that $G$ is decreasing) that there would exist a constant $c_3 > 0$ such that
	\begin{flalign*}
	(d + \lambda)^{-3/2} \exp \big( \Psi_{d, \lambda} + \lambda \big) & = (d + \lambda)^{-3/2} \exp \bigg( \lambda G \Big( \displaystyle\frac{d}{\lambda} + \lambda \Big) \bigg) \\
	& \ge (d + \lambda)^{-3/2} \exp \Big( \lambda \big( G (\mu') + 1 \big) \Big) \\
	& \ge (d + \lambda)^{-3/2} \exp \big( (c_3 - r^{-1} \lambda \big) = \delta (d + \lambda)^{-1} e^{(c_3 - o(1)) \lambda} > \delta.
	\end{flalign*}
	
	\noindent Thus, the lower bound in \Cref{pexponential2} implies that $d_{B; \delta} (e^x) \ge \big( \mu + o(1) \big) \lambda = \big( \mu r + o (1) \big) \log (\delta^{-1})$. 
	
	Similarly, if $d > \mu'' \lambda$ for some $\mu'' > \mu (r)$, then there exists some constant $c_4 > 0$ such that 
	\begin{flalign*}
	(d + \lambda)^2 \exp \big( \Psi_{d, \lambda} + \lambda \big) & \ge (d + \lambda)^2 \exp \Big( \lambda \big( G (\mu'') + 1 \big) \Big) \\
	& \ge (d + \lambda)^{-1} \exp \big( - \lambda (r^{-1} + c_4) \big) = \delta (d + \lambda)^2 e^{(o(1) - c_4) \lambda} < \delta, 
	\end{flalign*}
	
	\noindent which implies by the upper bound in \Cref{pexponential2} that $d_{B; \delta} (e^x) \le \big( \mu r + o (1) \big) \log (\delta^{-1})$. Hence, $d_{B; \delta} (e^x) = \big( \mu r + o (1) \big) \log (\delta^{-1})$.
	
	 Now let us consider the final case $B = o \big( \log (\delta^{-1}) \big)$. Suppose that 
	 \begin{flalign*}
	 d = \displaystyle\frac{\gamma \log (\delta^{-1}) }{\log (B^{-1} \log (\delta^{-1}))},
	 \end{flalign*}
	 
	 \noindent for some $\gamma \in (0, \infty)$ (bounded above and below). Then, $\frac{d}{\lambda} = \omega (1)$, and so the second part of \Cref{psivlambdakappa} implies that 
	 \begin{flalign*}
	 \Psi_{d, \lambda} = -d \log \Big( \displaystyle\frac{d}{\lambda} \Big) \big( 1 + o(1) \big).
	 \end{flalign*} 
	 
	 \noindent Hence, if $\gamma < 1$, then 
	 \begin{flalign*}
	(d + \lambda)^{-3/2} & \exp \big( \Psi_{d, \lambda} + \lambda \big) \\
	& = (d + \lambda)^{-3/2} \exp \bigg( -d \log \Big( \displaystyle\frac{d}{\lambda} \Big) \big( 1 + o(1) \big) \bigg) \\
	& = (d + \lambda)^{-3/2} \exp \bigg( -\displaystyle\frac{\gamma \log (\delta^{-1}) }{\log (B^{-1} \log (\delta^{-1}))} \log \Big( \displaystyle\frac{2 \gamma B^{-1} \log (\delta^{-1}) }{\log (B^{-1} \log (\delta^{-1}))} \Big) \big( 1 + o(1) \big) \bigg) \\
	& = (d + \lambda)^{-3/2} \exp \Big( \big(\gamma + o (1) \big) \log (\delta^{-1}) \Big) \ge  \delta^{\gamma + o(1)} \big( \log (\delta^{-1}) \big)^{-3/2} > \delta.
	 \end{flalign*}
    
    \noindent Hence, the lower bound in \eqref{pexponential2} implies that 
    \begin{flalign*}
	 d_{B; \delta} (e^x) \ge \displaystyle\frac{\big( 1 + o(1) \big) \log (\delta^{-1}) }{\log (B^{-1} \log (\delta^{-1}))},
	 \end{flalign*}
	 
	\noindent The proof of the matching upper bound is entirely analogous and is therefore omitted.
	\end{proof} 
	
	Next we establish \Cref{estimatedegree2}. To that end, we begin with the following lemma that addresses the last part of that theorem, when $B \ge \delta^{-\Omega (1)}$.

	\begin{lemma} 
	
	\label{degreeestimateb} 
	
	For every real $\delta \in (0,1/4)$ and $B \geq 1$ with $B > \omega(\log(\delta^{-1}))$ we have $d_{B;\delta}(e^{-x}) \geq (1/2 + o(1))\sqrt{B \log((2\delta)^{-1})}$.
	\end{lemma}
	
	\begin{proof}
	Let $p(x)$ be any polynomial satisfying $\sup_{x \in [0, B]} \big| p(x) - e^{-x} \big| < \delta$, and set $d = \deg p$. Let $x_0 = 0$, $x_1 = \log(\delta^{-1})$, and $x_2 = B$. It follows that $p(x_0) \geq 1-\delta$, and that $p(x) \in [-\delta,2\delta]$ for all $x \in [x_1,x_2]$. Let $a : \mathbb{R} \to \R$ be the linear function satisfying $a(1) = x_1$ and $a(-1) = x_2$, and let $$x_0' := a^{-1}(x_0) = 1 + \frac{2(x_1 - x_0)}{x_2 - x_1} = 1 + \frac{2\log(\delta^{-1})}{B-\log(\delta^{-1})}.$$ Finally, define the polynomial $q(x) = \frac{p(a(x))}{2\delta}$, which also has degree $d$.
    It follows that $q(x) \in [-1,1]$ for all $x \in [-1,1]$, and that $q(x_0') \geq \frac{1-\delta}{2\delta}$.
    
    Applying \Cref{extremalcheby} to $q$, we see that $|Q_d(x_0')| \geq 2^{1-d} q(x_0') \geq \frac{1- \delta}{2^d\delta}$. Furthermore, by \Cref{chebyasympt} we have that $|Q_d(x_0')| \leq 2^{-d}e^{(\sqrt{2} + o(1)) d \sqrt{x_0' - 1}}$. Combining the two bounds yields:
    $$\frac{1-\delta}{2^d\delta} \leq 2^{-d} e^{(\sqrt{2} + o(1)) d \sqrt{x_0' - 1}} = 2^{-d}e^{(2 + o(1)) d \sqrt{\log(\delta^{-1}) / (B-\log(\delta^{-1}))}}.$$
    Taking logs of both sides and rearranging gives the desired result.
	\end{proof}

    Now we can establish \Cref{estimatedegree2}.

	\begin{proof}[Proof of \Cref{estimatedegree2}]
	
	The proofs of the estimates on $d_{B; \delta} (e^{-x})$ in the first and second cases, when either $B = o \big( \log (\delta^{-1}) \big)$ and $B = \Theta \Big( \log (\delta^{-1}) \Big)$ are entirely analogous to those for $d_{B; \delta} (e^x)$ shown in \Cref{estimatedegree1} above. Therefore, they are omitted.
	
	So, let us assume that $B = \omega \big( \log (\delta^{-1}) \big)$, and let $d > 0$ be some integer with 
	\begin{flalign*}
	d = \sqrt{\gamma B \log (\delta^{-1})} = \sqrt{2 \gamma \lambda \log (\delta^{-1})},
	\end{flalign*}
	
	\noindent where $\gamma$ is uniformly bounded above and below. Observe since $B = \omega \big( \log (\delta^{-1}) \big)$ that $d = o (B)$, and so \Cref{psivlambdakappa} implies that 
	\begin{flalign*}
	\Psi_{d, \lambda} - \lambda = -\displaystyle\frac{d^2}{2 \lambda} \big( 1 + o(1) \big).
	\end{flalign*}
	
	Now, let us first approximate $d_{B; \delta} (e^{-x})$ by $( 1 + o(1) \big) \sqrt{B \log (\delta^{-1})}$ in the regime where $B \le \delta^{-o(1)}$. To lower bound it, suppose that $\gamma < 1$ (and is uniformly bounded away from $1$). Then, 
	\begin{flalign*}
	(\lambda + d)^{-3/2} \exp (\Psi_{d, \lambda} - \lambda) & \ge (\lambda + d)^{-3/2} \exp \Big( -\displaystyle\frac{d^2}{2\lambda} \big( 1 + o (1) \big) \Big) \\
	& \ge (\lambda + d)^{-3/2} \exp \Big( -\gamma \log (\delta^{-1}) \Big) \ge \delta^{\gamma + o(1)},
	\end{flalign*}
	
	\noindent where in the last bound we used the fact that $d = o (\lambda)$ and that $\lambda = \frac{B}{2} \le \delta^{-o(1)}$. Hence, by the lower bound in the second part of \Cref{pexponential2}, we find that $d_{B; \delta} (e^{-x}) \ge \big( 1 + o (1) \big) \sqrt{ B \log (\delta^{-1})}$.
	
	To upper bound $d_{B; \delta} (e^{-x})$ for $B \le \delta^{-o(1)}$, assume that $\gamma > 1$ (and is uniformly bounded away from $1$). Then,
	\begin{flalign*}
	\exp (\Psi_{d, \lambda} - \lambda) & \le (\lambda + d)^2 \exp \Big( -\displaystyle\frac{d^2}{2\lambda} \big( 1 + o (1) \big) \Big) \\
	& \le (\lambda + d)^2 \exp \Big( - \log (\delta^{-1}) \big( \gamma - o(1) \big) \Big) \le (\lambda + d)^2 \delta^{\gamma - o(1)} < \delta,
	\end{flalign*}
	
	\noindent and so again by \Cref{pexponential2} we deduce that $d_{B; \delta} (e^{-x}) \le \big( 1 + o(1) \big) \sqrt{B \log (\delta^{-1})}$. Together, these upper and lower bounds imply that $d_{B; \delta} (e^{-x}) = \big( 1 + o(1) \big) \sqrt{B \log (\delta^{-1})}$ when $B = \omega \big( \log (\delta^{-1}) \big)$ and $B \le \delta^{-o(1)}$.
	
	It remains to show that $d_{B; \delta} (e^{-x}) = \Theta \big( \sqrt{B \log (\delta^{-1})} \big)$ when $B \ge \delta^{-\Omega(1)}$. The lower bound (with implicit constant $\frac{1}{2} + o(1)$) was shown by \Cref{degreeestimateb}, so we must verify the upper bound. To that end, we assume $\gamma > 1$ (uniformly bounded away from 1) and observe that $B \ge 2d$ for $B \ge \delta^{-\Omega (1)}$. Then, the upper bound from \eqref{exponential3p} applies; since the first part of \Cref{psivlambdakappa}
	\begin{flalign*}
	\exp ( \Psi_{d, \lambda} - \lambda) \le \exp \Big( - \displaystyle\frac{d^2}{2 \lambda} \big( 1 + o(1) \big) \Big) \le \exp \Big( - \log (\delta^{-1}) \big( \gamma - o(1) \big) \Big) \le \delta^{\gamma - o(1)} \le \delta,
    \end{flalign*}
	
	\noindent we deduce that $d_{B; \delta} (e^{-x}) \le \big( 1 + o(1) \big) \sqrt{B \log (\delta^{-1}))}$, which establishes the theorem.
	\end{proof}

\section{Applications to Batch Gaussian KDE} \label{sec:apps}

In this section we prove the statements given in Section~\ref{section:gKDE} above about the Batch Gaussian KDE problem.

\begin{proof}[Proof of \Cref{kdecorr}]
Let $B \ge 1$ and $\delta \in (0, 1)$ denote real numbers, and suppose $p(z)$ is a univariate polynomial of degree $d = d_{B;\delta}(e^{-x})$ such that
\begin{flalign*}
\displaystyle\sup_{z \in [0, B]} \big| p(z) - e^{-z} \big| \le \delta.
\end{flalign*}
As discussed in the preamble to \Cref{kdecorr}, it suffices to output the vector $\tilde{K} \cdot w$, where $\tilde{K} \in \R^{n \times n}$ is the matrix given by $\tilde{K}[i,j] = p(\sum_{\ell=1}^m (x_\ell^{(i)} - y_\ell^{(j)})^2)$.

For two points $x, y \in \R^m$, we have that $p(\sum_{\ell=1}^m (x_\ell - y_\ell)^2)$ is a polynomial of degree at most $2d$ in the $2m$ variables in $V := \{x_1,\ldots,x_m,y_1,\ldots,y_m\}$. Thus, letting $M_d := \{ a : V \to \Z^{\geq 0} \mid \sum_{v \in V} a(v) \leq 2d \}$, which has $|M_d| = \binom{2m+2d}{2d} = M$, we can write \begin{align}\label{expandpoly2}p(\|x-y\|_2^2) = \sum_{a \in M_d} b_a \cdot \prod_{v \in V} v^{a(v)}\end{align}
for appropriate coefficients $b_a \in \R$ which can all be computed in $O\left( m \cdot \binom{2m+2d}{2d} \right)$ time by expanding $p$.

It follows that the matrix $\tilde{K}$, whose entries are given by the expression (\ref{expandpoly2}), has rank at most $|M_d| = M$, and furthermore that the matrices $X,Y \in \R^{n \times M}$ such that $\tilde{K} = X \times Y^T$ can be computed in $O\left( n \cdot m \cdot M \right)$ time by evaluating all the monomials in $M_G$ on the partial assignments of setting $x \leftarrow x^{(i)}$ for each $i \in [n]$ (to compute $X$) and setting $y \leftarrow y^{(j)}$ for each $j \in [n]$ (to compute $Y$). We can then, as desired, compute $\tilde{K}w = X(Y^T w)$ in time $O\left( n \cdot m \cdot M \right)$.
\end{proof}

Before proving \Cref{sethlb}, we give the necessary background about approximate nearest neighbor search. 

\begin{problem}[$(1+\eps)$-Approximate Hamming Nearest Neighbor]
For $\eps > 0$, and positive integers $n,m$, given as input vectors $a_1, \ldots, a_n, b_1, \ldots, b_n \in \{0,1\}^m$, as well as an integer $t \in [0,m]$, one must:
\begin{itemize}
    \item return `true' if there are $i,j \in [n]$ such that $|a_i - b_j| \leq t$,
    \item return `false' if, for every $i,j \in [n]$, we have $|a_i - b_j| > (1+\eps)\cdot t$,
\end{itemize}
and one may return either `true' or `false' otherwise. (Here, $|a_i - b_j|$ denotes the Hamming distance between $a_i$ and $b_j$.)
\end{problem}

\begin{theorem}[{\cite{rubinstein2018hardness}}] \label{aviadhardness}
Assuming SETH, for every $q>0$, there are $\eps>0$ and $C>0$ such that $(1+\eps)$-Approximate Hamming Nearest Neighbor in dimension $m = C \log n$ requires time $\Omega(n^{2-q})$.
\end{theorem}

\begin{proof}[Proof of \Cref{sethlb}]
For any $q>0$, let $\eps,C>0$ be the corresponding constants from \Cref{aviadhardness}. We will prove that Batch Gaussian KDE requires $\Omega(n^{2-q})$ time when $m = C \log n$, $\delta = n^{-2/\eps-1}/4$, and $B=2C \cdot c^{-1} \cdot \eps^{-1} \log n$ for a constant $c$ depending only on $q$ that we will determine later. We will prove this by showing that $(1+\eps)$-Approximate Hamming Nearest Neighbor in dimension $m$ can be solved using $o(n^{2-q})$ time and one call to Batch Gaussian KDE with these parameters, which implies the desired result when combined with \Cref{aviadhardness}.

Let $m = C \log n$, let $a_1, \ldots, a_n, b_1, \ldots, b_n \in \{0,1\}^m$ be the input vectors to $(1+\eps)$-Approximate Hamming Nearest Neighbor, and let $t \in [0,m]$ be the target distance. First, if $t < c \log n$, we will simply brute-force for the answer in the following way: we store the vectors $b_1, \ldots, b_n$ in a lookup table, then for each $i \in [n]$, we iterate over every vector $b' \in \{0,1\}^m$ which has Hamming distance at most $t$ from $a_i$ and check whether it is in the lookup table. The running time will be only $O\left( n \cdot \binom{m}{t} \right)$ when $c$ is small enough. In particular,
$$\binom{m}{t} \leq \binom{C \log n}{c \log n} \leq n^{f(C,c)}$$
for some function $f : \R_{>0} \times \R_{>0} \to \R_{>0}$ with the property that, for any fixed $C>0$, we have $\lim_{c \to 0}f(C,c) = 0$. We can thus pick a sufficiently small constant $c>0$, depending only on $q$ and $C$ (which itself depends only on $q$) such that this entire brute-force takes $o(n^{2-q})$ time.

Henceforth, we assume that $t = c' \log n$ for some $C \geq c' > c$. Let $k = \sqrt{2(c' \eps)^{-1}}$.
Using our algorithm for Batch Gaussian KDE with the given parameters applied to the input points $ka_1, \ldots, ka_n, kb_1, \ldots, kb_n$  (i.e., the input points rescaled so they lie in $\{0,k\}^m$), and the weight vector $w = \vec{1}\in \R^n$, the all-1s vector, we get as output a vector $v \in \R^n$ such that, for all $i \in [n]$, we have the guarantee that $$\bigg| v_i - \sum_{j=1}^n e^{-k^2 \cdot \| x_i - y_j \|_2^2} \bigg| < n^{-1/\eps}/4.$$
Notice that, for $x_i, y_j \in \{0,1\}^m$, the quantity $\| x_i - y_j \|_2^2$ is equal to the Hamming distance $|x_i - y_j|$ between $x_i$ and $y_j$. In particular, we have $\max_{i,j \in [n]} \|kx_i - ky_j\|_2^2 \leq k^2 m = (2C c'^{-1} \eps^{-1} ) \log n \leq (2C c^{-1} \eps^{-1}) \log n = B$, so this was a valid application of our given algorithm.

Suppose first that there are an $i,j \in [n]$ such that $|x_i - y_j| \leq t$. It follows that $v_i \geq e^{-k^2t} - n^{-2/\eps}/4 = \frac34 n^{-2/\eps}$.

Suppose second that, for all $i,j \in [n]$, we have $|x_i - y_j| > t(1+\eps)$. It follows that, for all $i \in [n]$, we have $v_i \leq n \cdot e^{-k^2t(1+\eps)} + n^{-1/\eps}/4 = n^{-2/\eps-1} + n^{-2/\eps}/4 < \frac34 n^{-2/\eps}$ for all $n>2$.

Hence, by checking whether any entry of $v$ is at least $\frac34 n^{-1/\eps}$, we can distinguish the two cases, as desired.
\end{proof}

We conclude with the proof of \Cref{fkappa}. 

\begin{proof}[Proof of \Cref{fkappa}]

By \Cref{kdecorr}, it suffices to show that there exists a constant $c = c(\alpha, \beta) > 0$ such that 
\begin{flalign} 
\label{dm}
\binom{2d + 2m}{2d} \le n^{c \log \log \kappa^{-1} / \log \kappa^{-1}},
\end{flalign} 

\noindent for $m = \alpha \log n$ and $d = d_{B; \delta} (e^{-x})$, where $\delta = n^{-\beta}$ and $B = \kappa \log n$. In particular, $B = 2 r \log (\delta^{-1})$, where $r = \frac{\kappa}{2 \beta}$. Thus, \Cref{estimatedegree2} gives
\begin{flalign}
\label{dkappa} 
2 d = \big(\kappa \nu + o(1) \big) \log n,
\end{flalign}

\noindent where $\nu = \nu \big( \frac{\kappa}{2 \beta} \big) > 0$ is the unique positive solution to $G (\nu) = 1 - r^{-1} = 1 - \frac{2 \beta}{\kappa}$ (where $G$ is given by \eqref{gx}). 

To establish \eqref{dm}, we must analyze the behavior of $d$ and thus of $\nu$. We claim that
\begin{flalign}
\label{nuestimate}
\nu = \displaystyle\frac{2 \beta}{\kappa \log \kappa^{-1}} + O \bigg( \displaystyle\frac{\log \log \kappa^{-1}}{\kappa (\log \kappa^{-1})^2} \bigg).
\end{flalign}

Let us quickly establish the corollary assuming \eqref{nuestimate}. To that end, denote
 \begin{flalign*} 
 x = \kappa \nu = \displaystyle\frac{2 \beta}{\log \kappa^{-1}} + O \bigg( \displaystyle\frac{\log \log \kappa^{-1}}{(\log \kappa^{-1})^2}\bigg).
 \end{flalign*} 
 
 \noindent Then, \eqref{dkappa} and a Taylor expansion together give
 \begin{flalign*}
 \log \binom{2m + 2d}{2d} & = \big( (2\alpha + x) \log (2 \alpha + x) - 2 \alpha \log (2 \alpha) - x \log x + o (1) \big) \log n \\
 & = \big( x + x \log (2 \alpha + x) - x \log x + O (x^2) \big) \log n = x |\log x| \Bigg( 1 + O \bigg( \displaystyle\frac{1}{|\log x|} \bigg) \Bigg) \log n,
 \end{flalign*} 
 
 \noindent where the implicit constant in the error depends on $\alpha$ and $\beta$. Thus,
 \begin{flalign*} 
 \binom{2 m + 2 d}{d} \le n^{x |\log x| + O (x)} = n^{2 \beta \log \log \kappa^{-1} / \log \kappa^{-1} + O(1/ \log \kappa^{-1})},
 \end{flalign*}
 
 \noindent which verifies \eqref{dm} and thus the corollary.

It thus remains to verify \eqref{nuestimate}. To that end, observe that 
\begin{flalign}
\label{kappa1} 
\displaystyle\frac{2 \beta}{\kappa} = 1 - G (\nu) = 1 - \sqrt{\nu^2 + 1} - \nu \log \big( \sqrt{\nu^2 + 1} - \nu \big) = \nu \log \big( \sqrt{\nu^2 + 1} + \nu \big) - \sqrt{\nu^2 + 1} + 1.
\end{flalign}

\noindent In particular, this implies that
\begin{flalign*}
\nu \log \big( \sqrt{\nu^2 + 1} + \nu \big) \ge \displaystyle\frac{2 \beta}{\kappa}, \qquad \text{so $\nu = \Omega \bigg( \displaystyle\frac{1}{\kappa \log \kappa^{-1}} \bigg)$},
\end{flalign*}

\noindent where the implicit constant depends on $\beta$. The above lower bound enables us to Taylor expand the right side of \eqref{kappa1}. This gives
\begin{flalign*}
\nu \log \nu + \nu (\log 2 - 1) + O(\nu^{-1}) = \displaystyle\frac{2 \beta}{\kappa},
\end{flalign*} 

\noindent from which \eqref{nuestimate} quickly follows. As mentioned above, this verifies the corollary.
\end{proof}

\section*{Acknowledgements} 
 We would like to thank Alexei Borodin, Lijie Chen, Andrei Martinez-Finkelshtein, Sushant Sachdeva, Paris Siminelakis, Roald Trigub, Ryan Williams, and anonymous reviewers for helpful discussions and advice throughout this project.

\bibliographystyle{alpha}
\bibliography{papers}

\end{document}